%% file: mergeres-final-source.tex
\PassOptionsToPackage{dvipsnames}{xcolor}

\documentclass[manuscript,screen,nonacm]{acmart}

\AtBeginDocument{%
	\providecommand\BibTeX{{%
			\normalfont B\kern-0.5em{\scshape i\kern-0.25em b}\kern-0.8em\TeX}}}

\acmJournal{TOCT}

\input{notation-acm}

\input{JLBtikz.tex}

\begin{document}
\title{Hard QBFs for Merge Resolution}
\titlenote{A preliminary version of this article appeared in the proceedings of the 40th IARCS Annual Conference on Foundations of Software Technology and Theoretical Computer Science -- FSTTCS 2020 \cite{BBMPS-FSTTCS20}.}

\author{Olaf Beyersdorff}
\orcid{0000-0002-2870-1648}
\email{olaf.beyersdorff@uni-jena.de}
\affiliation{%
	\institution{Institut f\"{u}r Informatik, Friedrich-Schiller-Universit\"{a}t Jena}
	\streetaddress{Ernst-Abbe-Platz 2, 07743}
	\city{Jena}
	\country{Germany}}

\author{Joshua Blinkhorn}
\orcid{0000-0001-7452-6521}
\email{joshua.blinkhorn@uni-jena.de}
\affiliation{%
	\institution{Institut f\"{u}r Informatik, Friedrich-Schiller-Universit\"{a}t Jena}
	\streetaddress{Ernst-Abbe-Platz 2, 07743}
	\city{Jena}
	\country{Germany}}

\author{Meena Mahajan}
\orcid{0000-0002-9116-4398}
\email{meena@imsc.res.in}
\affiliation{%
	\institution{The Institute of Mathematical Sciences  (CI of Homi Bhabha National Institute)}
	\streetaddress{IV Cross Road, CIT Campus, Taramani}
	\city{Chennai}
	\country{India}}

\author{Tom\'{a}\v{s} Peitl}
\orcid{0000-0001-7799-1568}
\email{peitl@ac.tuwien.ac.at}
\affiliation{%
	\institution{Institute of Logic and Computation, TU Wien}
	\streetaddress{Favoritenstraße 9-11, 1040}
	\city{Vienna}
	\country{Austria}}

\author{Gaurav Sood}
\authornote{Current affiliation: Department of Computer Science, University of Haifa, Israel}
\orcid{0000-0001-6501-6589}
\email{gaurav.sood.work@gmail.com}
\affiliation{%
	\institution{The Institute of Mathematical Sciences  (CI of Homi Bhabha National Institute)}
	\streetaddress{IV Cross Road, CIT Campus, Taramani}
	\city{Chennai}
	\country{India}}
	
\renewcommand{\shortauthors}{O.~Beyersdorff, J.~Blinkhorn, M.~Mahajan, T.~Peitl, and G.~Sood}	

\begin{abstract}
	We prove the first genuine QBF proof size lower bounds for the proof system Merge Resolution (MRes \cite{DBLP:journals/jar/BeyersdorffBM21}), a refutational proof system for prenex quantified Boolean formulas (QBF) with a CNF matrix. Unlike most QBF resolution systems in the literature, proofs in MRes consist of resolution steps \emph{together} with information on countermodels, which are syntactically stored in the  proofs as merge maps. As demonstrated in  \cite{DBLP:journals/jar/BeyersdorffBM21}, this makes MRes quite powerful: it has strategy extraction by design and allows short proofs for formulas which are hard for classical QBF resolution systems.
	
	Here we show the first genuine QBF \emph{exponential lower bounds for MRes}, thereby uncovering limitations of MRes. Technically, the results are either transferred from bounds from circuit complexity (for restricted versions of MRes) or directly obtained by combinatorial arguments (for full MRes).
	Our results imply that the MRes approach is \emph{largely orthogonal to other QBF resolution models} such as the QCDCL resolution systems QRes and QURes and the expansion systems $\ExpRes$ and $\IR$.
\end{abstract}



\maketitle

\section{Introduction}
\label{sec:intro}
\emph{Proof complexity} aims to provide a theoretical understanding of the
ease or difficulty of proving statements formally. It also aims to
explain the success stories of, as well as the obstacles faced by,
algorithmic approaches to hard problems such as satisfiability (SAT) and
Quantified Boolean Formulas (QBF) \cite{Bus12,Nordstrom15}. While propositional proof complexity,
the study of proofs of unsatisfiability of propositional formulas, has
been around for decades \cite{CR79,Kra95}, the area of \emph{QBF proof complexity} is 
relatively new, with theoretical studies gaining traction only in the
last decade or so \cite{BWJ-SAT14,BBCP20,BCJ-ToCT-Sep2019,BBMP-ToCL23,M4CQBF}. While inheriting and using a wealth of techniques
from propositional proof complexity \cite{BCS19,BCMS17,KauersS18}, QBF proof complexity has also
given several new perspectives specific to QBF \cite{SlivovskyS16,BBH-LMCS19,Janota-Expansion-vs-QRes-TCS15}, and these
perspectives and their connections to QBF solving \cite{ZM02,PulinaS19,sathandbookqbf,BB23-LMCS} as well as their practical applications \cite{ShuklaBPS19} have driven the search for newer proof systems \cite{HSB14,PeitlSS19,BCJ-ToCT-Sep2019,Balabanov12,LonsingES16}.

Many of the currently known QBF proof systems are built on the best-studied propositional proof system \emph{resolution} \cite{Bla37,Rob65}.  Broadly speaking,
resolution has been adapted to handle and eliminate the universal variables in QBFs
in two intrinsically different ways. The first is an \emph{expansion-based
approach:} universal variables are eliminated  by
implicitly expanding the universal quantifiers into conjunctions,
creating annotated copies of existential variables. Universal variables thus appear in the proofs only in the annotations. The systems
$\ExpRes$, $\IR$, and $\IRM$
\cite{Janota-Expansion-vs-QRes-TCS15,BCJ-ToCT-Sep2019} are of this
type. The second is a \emph{reduction-rule approach:} under certain
conditions, resolution may be blocked, and also under certain
conditions, universal variables can be deleted from clauses. The
conditions are formulated to preserve soundness, ensuring that if a
QBF is true, then so is the QBF resulting from adding a derived
clause.  The systems $\QRes$, $\QURes$, $\QCP$
\cite{KBKF95,Gelder12,BCMS-IC18} are of this type.

A central role in QBF proof complexity is played by the \emph{two-player
evaluation game} on QBFs, and the existence of winning strategies for
the universal player in false QBFs.  For many QBF resolution systems,
such strategies were used to construct proofs and demonstrate
completeness, and soundness was demonstrated by extracting such
strategies from proofs \cite{ELW13,Balabanov12,BCJ-ToCT-Sep2019}.  The \emph{strategy extraction} procedures build
partial strategies at each line of the proof, with the strategies at
the final line forming a complete countermodel. These extraction
procedures are based on the fact that in each application of a rule in
the proof system, any winning strategies of the existential player are
not destroyed.

In the systems $\QRes$ \cite{KBKF95} and $\QURes$ \cite{Gelder12}, the soundness of the resolution rule
is ensured by enforcing a very simple side-condition: variables other
than the resolved variable (referred to henceforth as the pivot) cannot appear in both polarities in the
antecedents. It was observed early on that this is often too
restrictive. The \emph{long-distance resolution proof system} $\LDQRes$
\cite{Balabanov12,ZM02} arose from efforts to have less restrictive but
still sound rules. In this system, a universal variable could appear
in both polarities and get merged in the consequent, provided it was
to the right of the pivot in the quantifier prefix. This preserves
soundness, but the strategy extraction procedures become notably more
complex.

The system $\LDQRes$, while provably better than $\QRes$ \cite{ELW13}, is still
needlessly restrictive in some situations. In particular, by checking
a very simple syntactic prefix-ordering condition, it fails to exploit
the fact that soundness is not lost even if universal variables to the
left of the pivot are merged in both antecedents, provided the partial
strategies built for them in both antecedents are identical.  A \emph{new
system Merge Resolution} (\emph{$\MRes$}) was introduced recently \cite{DBLP:journals/jar/BeyersdorffBM21}
by a subset of the current authors, precisely to address this
point. In $\MRes$, partial strategies are explicitly represented
within the proof, in a particular representation format called merge
maps -- these are essentially deterministic branching programs
(DBPs). In this format, isomorphism checking can be done efficiently,
and this opens the way for enabling sound applications of resolution
that would have been blocked in $\LDQRes$ (and $\QRes$).  In \cite{DBLP:journals/jar/BeyersdorffBM21}, it
was shown that this permitted a simulation of
reductionless $\LDQRes$, denoted $\rLDQRes$, the fragment of $\LDQRes$ where all reductions are postponed to the very end (no reduction step is followed by a resolution step). (This fragment was identified as interesting in \cite{BjornerJK15}; see also \cite{PeitlSS19a}.) More importantly, it was also shown in \cite{DBLP:journals/jar/BeyersdorffBM21} that enabling resolution steps blocked in $\LDQRes$ 
brought a rich pay-off: there are families of
formulas, the Equality and the SquaredEquality formulas, with short (linear-size)
proofs in $\MRes$, even in its tree-like and regular versions, but
requiring exponential size in $\QRes$, $\QURes$, $\QCP$,
$\ExpRes$, and $\IR$. It is notable that the hardness of
Equality (and also Squared Equality) in these systems stems from a certain semantic cost
associated with these formulas and a corresponding lower bound
\cite{BBH-LMCS19, BB-TOCS20}. Thus the results of \cite{DBLP:journals/jar/BeyersdorffBM21}
show that such semantic costs are not a barrier for $\MRes$.

In this paper, we explore the price paid for overcoming the semantic
cost barrier. We show that (expectedly) $\MRes$ is not an unconditional %
success story. Building strategies into proofs via merge maps, and
screening out unsoundness only through isomorphism tests, comes at a
fairly heavy price: exponentially long proofs for various formulas. 

It may be noted that for existentially quantified QBFs, all the QBF proof systems mentioned in this paper coincide with Resolution (or in case of $\QCP$, with Cutting Planes). Therefore lower bounds for these propositional proof systems trivially lift to the corresponding QBF proof system. In particular, the separations of
tree-like and regular $\MRes$ from $\MRes$ and other systems follow directly from the propositional case. However, such lower bounds do not tell us much about the limitations of the QBF proof system other than what is known from the underlying propositional proof system. Therefore, in QBF proof complexity, we are interested in `genuine' QBF lower bounds, i.e.~lower bounds that do not follow from propositional lower bounds (cf.\ \cite{BHP20} on how to formally define the notion of `genuine' lower bounds in many QBF proof systems). The lower bounds we establish here are of this nature. Specifically, we may consider an  $\MRes$ derivation of a line from a given set of lines to be purely propositional if at each step, each merge map  appearing in the consequent line already appears (in an isomorphic form) in at least one of the antecedent lines. The derivation thus does not contribute to building up the strategies. Collapsing such derivations to  single steps (say, by accessing an NP oracle) leaves behind a proof in which purely propositional hardness has been removed. Our arguments show that even such proofs must be large; in this sense, our bounds are genuine QBF lower bounds. 

\begin{figure}[ht]
	\begin{center}
		\begin{tikzpicture}[>=triangle 45,scale=0.94, transform shape]
			
			\newcommand{\TopRank}{4.9}
			\newcommand{\MiddleRank}{3.5}
			\newcommand{\BottomRank}{2.1}
			\newcommand{\MergeTopRank}{4.9}
			\newcommand{\MergeBottomRank}{2.4}
			\newcommand{\MergeFile}{0.6}
			\newcommand{\CentreFile}{6.1}
			\newcommand{\LDQFile}{11.1}
			\newcommand{\MergeOffset}{1.0}
			\newcommand{\CentreOffset}{1.2}
			
			\newcommand{\MergeRecWidth}{2.0}
			\newcommand{\CenterRecWidth}{2.5}
			\newcommand{\LDRecWidth}{1.4}
			\newcommand{\RecHeight}{0.6}
			\newcommand{\MergeRecHeight}{0.8}
			\newcommand{\CenterRecHeight}{0.6}
			\newcommand{\LDRecHeight}{0.8}
			\newcommand{\Offset}{0.1} %
			\draw[fill=black!5, rounded corners = 0.15cm] (\CentreFile - \CenterRecWidth, \TopRank + \CenterRecHeight) rectangle (\CentreFile + \CenterRecWidth + \Offset, \BottomRank - \CenterRecHeight);
			\draw[fill=black!5, rounded corners = 0.15cm] (\MergeFile - \MergeRecWidth, \MergeBottomRank + \MergeRecHeight) rectangle (\MergeFile + \MergeRecWidth, \MergeBottomRank - \MergeRecHeight);
			\draw[fill=black!5, rounded corners = 0.15cm] (\LDQFile - \LDRecWidth, \MergeTopRank + \CenterRecHeight) rectangle (\LDQFile + \LDRecWidth, \MergeBottomRank - \LDRecHeight);
			
			\node[draw, rounded corners](M) at (\MergeFile,\MergeTopRank){$\MRes$};
			\node[draw, rounded corners, align=center](RM) at (\MergeFile + \MergeOffset,\MergeBottomRank){Regular\\$\MRes$};
			\node[draw, rounded corners, align=center](TM) at (\MergeFile - \MergeOffset,\MergeBottomRank){Tree-like\\$\MRes$};
			\node[draw, rounded corners](CP) at (\CentreFile - \CentreOffset,\TopRank){$\QCP$};
			\node[draw, rounded corners](QU) at (\CentreFile - \CentreOffset,\MiddleRank){$\QURes$};
			\node[draw, rounded corners](Q) at (\CentreFile - \CentreOffset,\BottomRank){$\QRes$};
			\node[draw, rounded corners](LDQ) at (\LDQFile,\MergeTopRank){$\LDQRes$};
			\node[draw, rounded corners](IR) at (\CentreFile + \CentreOffset,\TopRank){$\IR$};
			\node[draw, rounded corners](E) at (\CentreFile + \CentreOffset,\MiddleRank){$\ExpRes$};
			\node[draw, rounded corners, align=center](rLDQ) at (\LDQFile,\MergeBottomRank){Reductionless\\$\LDQRes$};
			
			\newcommand{\Thick}{0.3mm}
			\newcommand{\Thicker}{0.55mm}
			
			\node (P1) at (\MergeFile + \MergeRecWidth - 0.05, \MergeBottomRank){};
			\node (P2) at (\CentreFile - \CenterRecWidth + 0.08, \MergeBottomRank){};
			\node (P3) at (\CentreFile + \CenterRecWidth, \MiddleRank){};
			\node (P4) at (\LDQFile - \LDRecWidth + 0.08, \MiddleRank){};
			\draw[dotted, line width = \Thicker](M) -- (CP) ;
			\draw[dotted, line width = \Thicker](M) -- (QU) ;
			\draw[dotted, line width = \Thick](P1) -- (P2) ;
			\draw[dotted, line width = \Thick, gray](P3) -- (P4) ;

			\begin{scope}[decoration = {markings, mark = at position 0.7 with {\arrow{>}}}]
				\draw[postaction = {decorate}, gray](CP) -- (QU);
				\draw[postaction = {decorate}, gray](QU) -- (Q);
				\draw[postaction = {decorate}, gray](IR) -- (E);
				\draw[postaction = {decorate}, gray](LDQ) to [out = 240, in = 1] (Q);
				\draw[postaction = {decorate}, gray](IR) to [out = 200, in = 25] (Q);
				\draw[postaction = {decorate}, gray](LDQ) -- (rLDQ);
				\draw[postaction = {decorate}, gray](M) to [out = 45, in = 105] (rLDQ);
			\end{scope}
			
			\begin{scope}[decoration = {markings, mark = at position 0.6 with {\arrow[gray]{>}}}]
				\draw[postaction = {decorate}, gray] (M) -- (TM);
				\draw[postaction = {decorate}, gray, dotted, line width = \Thick] (M) -- (RM);
			\end{scope}
			
			\newcommand{\KeyRank}{0}
			\newcommand{\KeyLeft}{0.1}
			\newcommand{\KeyCentre}{2.1}
			\newcommand{\KeyRight}{6.9}
			\newcommand{\KeyOffset}{0.5}
			\newcommand{\AtoBOffset}{1.8}
			\newcommand{\DescriptionOffset}{4.3}
			
			\node[draw, align=center, rounded corners=0.05cm](Ai) at (\KeyLeft,\KeyRank + \KeyOffset){$\mathsf{A}$};
			\node[draw, align=center, rounded corners=0.05cm](Bi) at (\KeyLeft + \AtoBOffset,\KeyRank + \KeyOffset){$\mathsf{B}$};      
			\begin{scope}[decoration = {markings, mark= at position 0.6 with {\arrow{>}}}]
				\draw[postaction = {decorate}, dotted, line width = \Thick](Ai)--(Bi);      
			\end{scope}
			\node[align = left, text width = 4cm] at (\KeyLeft + \DescriptionOffset,\KeyRank + \KeyOffset){$\mathsf{A}$ $p$-simulates $\mathsf{B}$};
			
			\node[draw, align=center, rounded corners=0.05cm](Ai) at (\KeyLeft,\KeyRank - \KeyOffset){$\mathsf{A}$};
			\node[draw, align=center, rounded corners=0.05cm](Bi) at (\KeyLeft + \AtoBOffset,\KeyRank - \KeyOffset){$\mathsf{B}$};      
			\begin{scope}[decoration = {markings, mark= at position 0.6 with {\arrow{>}}}]
				\draw[postaction = {decorate}, line width = \Thick](Ai)--(Bi);      
			\end{scope}
			\node[align = left, text width = 4cm] at (\KeyLeft + \DescriptionOffset,\KeyRank - \KeyOffset){$\mathsf{A}$ \emph{strictly} \\ $p$-simulates $\mathsf{B}$};
			
			\node[draw, align=center, rounded corners=0.05cm](Aii) at (\KeyRight,\KeyRank){$\mathsf{A}$};
			\node[draw, align=center, rounded corners=0.05cm](Bii) at (\KeyRight + \AtoBOffset,\KeyRank){$\mathsf{B}$};
			\draw[dotted, line width = \Thick](Aii)--(Bii);
			\node[align = left, text width = 4cm] at (\KeyRight + \DescriptionOffset,\KeyRank){$\mathsf{A}$ and $\mathsf{B}$ are\\incomparable};
			
		\end{tikzpicture}
	\end{center}
	\caption{Visual summary of the proof complexity landscape, with new results shown in bold. %
	Lines from/to a big grey box mean that the line is from/to every proof system within the box.
	New separations are summarised in \Cref{cor:tree-and-regular-MRes-incomparable-with-five-systems,cor:MRes-incomparable}.
	}
	\label{fig:QBF-simulation-order}
	\Description{}
\end{figure}

\smallskip\noindent
\textbf{(A) Lower bounds from circuit complexity for restricted versions of $\MRes$.}~Since the strategies are explicitly represented inside
the proofs, computational hardness of strategies immediately
translates to proof size lower bounds. While computational hardness of
strategies is a known source of hardness in all reduction-based proof
systems admitting efficient strategy extraction \cite{BBC-ITCS16,BCJ-ToCT-Sep2019},
the computational model relevant for $\MRes$ is
one for which no unconditional lower bounds are known. For tree-like
and regular $\MRes$, the relevant models are decision trees and
read-once DBPs, where lower bounds are known. Using this approach, we
show:
\begin{enumerate}
\item Tree-like $\MRes$ does not simulate regular and general $\MRes$,  in terms of genuine size.
  The  QParity formulas witness the separation
  (\Cref{thm:treeMRes-parity-lb}) as their unique countermodel is the parity function which requires large decision trees.

Note: unlike in the propositional setting, we do not know whether regular $\MRes$ simulates tree-like $\MRes$. 
\item Tree-like $\MRes$ is incomparable with the dag-like and tree-like
  versions of $\QRes$, $\QURes$, $\QCP$, $\ExpRes$ and $\IR$ (\Cref{cor:tree-and-regular-MRes-incomparable-with-five-systems}). \\
  One direction was shown in \cite{DBLP:journals/jar/BeyersdorffBM21} via the Equality
  formulas: these formulas are easy for tree-like $\MRes$ but hard for dag-like $\QRes$, $\QURes$, $\QCP$, $\ExpRes$, $\IR$. The other direction is witnessed by the Completion
  Principle formulas, easy in tree-like versions of $\QRes$ and $\ExpRes$ \cite{Janota-QRes-and-CDCL-SAT16,Janota-Expansion-vs-QRes-TCS15}, but exponentially hard for tree-like $\MRes$ (\Cref{thm:treeMRes-CR}). Unlike the
  QParity formulas, these formulas do not have unique
  countermodels. However, we show that every countermodel requires
  large decision-tree size, and hence obtain the lower bound for tree-like $\MRes$.%
\end{enumerate}

\noindent \textbf{(B) Combinatorial lower bounds for $\MRes$.}~Even when winning strategies are easy to
compute by DBPs, the formulas can be hard for $\MRes$.  We establish
such hardness in three cases, obtaining more incomparabilities.
\begin{enumerate}
\item The LQParity formulas, %
  easy in $\ExpRes$
  \cite{BCJ-ToCT-Sep2019}, are exponentially hard for regular $\MRes$
  (\Cref{thm:LQParity-lb}). Hence regular $\MRes$ is incomparable with
  $\ExpRes$ and~$\IR$.%
  \item The Completion Principle formulas, easy in tree-like versions of $\QRes$ and $\ExpRes$ \cite{Janota-QRes-and-CDCL-SAT16,Janota-Expansion-vs-QRes-TCS15}, are exponentially hard for regular $\MRes$ (\Cref{thm:CR-regular-lower-bound}). Hence regular $\MRes$ is incomparable with the dag-like and tree-like versions of $\QRes$, $\QURes$, $\QCP$, $\ExpRes$ and $\IR$ (\Cref{cor:tree-and-regular-MRes-incomparable-with-five-systems}).
\item The KBKF-lq formulas, %
  easy in $\QURes$
  \cite{BWJ-SAT14}, are exponentially hard for $\MRes$
  (\Cref{thm:mres-lb}). Hence $\MRes$ is incomparable with $\QURes$ and $\QCP$ (\Cref{cor:MRes-incomparable}). 
\end{enumerate}

The third hardness result above for the KBKF-lq formulas provides the first genuine lower bound for the full system of $\MRes$, for which previously no such lower bounds were known.

\Cref{fig:QBF-simulation-order} depicts the \emph{simulation order and
incomparabilities} we establish involving $\MRes$ and its
refinements. Amongst the five systems
in the big grey box, all relationships not directly implied by depicted
connections are known to be incomparabilities \cite{BCJ-ToCT-Sep2019,Janota-Expansion-vs-QRes-TCS15,BCMS-IC18}.

More recently, upper bounds for  the proof system $\MRes$ have been  established, in \cite{CS-STACS22}, and variants of $\MRes$  have been explored, in \cite{CS-STACS23}. 

\paragraph{Organization of this paper} We define QBFs and $\MRes$ in \Cref{sec:prelim}. In \Cref{sec:lower bounds}, we prove lower bounds for many formula families. Finally, in \Cref{sec:proof-system-realtions}, we give the resulting separations among QBF proof systems.

\section{Preliminaries}
\label{sec:prelim}
Let $[n] = \{ 1,2, \ldots, n\}$ and $[m,n] = \{m, \ldots,n\}$.

Variables take Boolean values, and a literal $\ell$ is a variable $x$ or its negation $\neg x$ (also denoted $\bar{x}$). We say that $x=\var(\ell)$. A clause is a disjunction of literals, and a conjunctive-normal-form (CNF) formula is a conjunction of clauses.  
We represent clauses interchangeably as disjunctions of literals and sets of literals. Similarly, we represent CNF formulas interchangeably as conjunctions of clauses and sets of clauses. 

The \emph{resolution rule} derives, from clauses $C\vee \ell$ and $D\vee \neg
\ell$ for some literal $\ell$, the clause $C\vee D$.  We say that $C \vee D$ is the resolvent, $x=\var(\ell)$ is
the pivot, and denote this by $C \vee D = \res(C\vee \ell, D\vee \neg \ell, x)$. Representing clauses as sets of literals, we say that $C \cup D$ is the resolvent of $C \cup \{\ell\}$ and $D \cup \{\overline{\ell}\}$ on pivot $x$, and denote this by $C \cup D = \res(C \cup \{\ell\}, D \cup \{\overline{\ell}\}, x)$. 

The \emph{propositional proof system Resolution} proves that a CNF formula
$F$ is unsatisfiable by deriving the empty clause through repeated
applications of the resolution rule.

\subsection{Quantified Boolean formulas}
A \emph{Quantified Boolean Formula} (QBF) in \emph{prenex conjunctive normal form} is denoted $\Phi \coloneqq Q \cdot \phi$, where 
\begin{itemize}
	\item $Q = Q_1 Z_1 Q_2 Z_2 \ldots Q_k Z_k$ is the quantifier prefix, in which $Z_i$ are pairwise disjoint finite sets of Boolean variables, $Q_i \in \{\exists,\forall\}$ for each $i \in [k]$ and $Q_i \neq Q_{i+1}$ for each $i \in [k-1]$, and
	\item the matrix $\phi$ is a CNF over $\vars(\Phi) \coloneqq \cup_{i \in [k]} Z_i$.
\end{itemize}

The existential (resp.~universal) variables of $\Phi$, typically denoted $X$ or $X_\exists$ (resp.~$U$ or $X_\forall$) is the set obtained as the union of $Z_i$ for which $Q_i = \exists$ (resp.~$Q_i = \forall$). The prefix $Q$ defines a binary relation $<_Q$ on $\vars(\Phi)$, such that $z <_Q z'$ holds iff $z \in Z_i$, $z' \in Z_j$, and $i < j$, in which case we say that $z'$ is right of $z$ and $z$ is left of $z'$. For each $u \in U$, we define $L_Q(u) \coloneqq \{x \in X \mid x <_Q u\}$, i.e.~the existential variables left of $u$.

For a set of variables $Z$, let $\langle Z \rangle$ denote the set of assignments to $Z$. A \emph{strategy} $h$ for a QBF $\Phi$ is a set $\{h^u \mid u \in U\}$ of functions $h^u \colon \langle L_Q(u) \rangle \to \{0,1\}$ (for each $\alpha \in \langle X \rangle$, $h^u(\alpha \restriction_{L_Q(u)})$ and $h(\alpha)$ should be interpreted as a Boolean assignment to the variable $u$ and the variable set $U$ respectively). Additionally $h$ is \emph{winning} if, for each $\alpha \in \langle X \rangle$, the restriction of $\phi$ by the assignment $(\alpha, h(\alpha))$ is false. We use the terms ``winning strategy'' and ``countermodel'' interchangeably. A QBF is called false if it has a countermodel, and true if it does not.

The semantics of QBFs is also explained by a \emph{two-player evaluation
game} played on a QBF. In a run of the game, two players, the
existential and the universal player, assign values to the variables
in the order of quantification in the prefix. The existential player
wins if the assignment so constructed satisfies all the clauses of
$\phi$; otherwise the universal player wins. Assigning values
according to a countermodel guarantees that the universal player wins
no matter how the existential player plays; hence the term ``winning
strategy''.

\subsection{The  Merge Resolution proof system}

We first describe the idea behind the Merge Resolution ($\MRes$) proof system.
\MRes\ is a line-based proof system. A refutation in Merge Resolution is a sequence of lines. Each line $L$ consists of a clause $C$
with only existential literals, and a partial strategy $h^u$ for each
universal variable $u$. The idea is to maintain the invariant that for
each existential assignment $\alpha$, if $\alpha$ falsifies $C$, then
$\alpha$ extended by the partial universal assignment setting each $u$
to $h^u(\alpha)$ falsifies at least one of the clauses used to derive
$L$.  Thus the set of functions $\{h^u\}$ gives a partial strategy (for the universal player)
that wins whenever the existential player plays from the set of
assignments falsifying $C$. The goal is to derive a line with the
empty clause; the corresponding strategy at that line will be a
complete winning strategy for the universal player, i,e.~a countermodel. Along the way, resolution
is used on the clauses. If the pivot is $x$, then for universal
variables $u$ right of $x$, the partial strategies can be combined
with a branching decision on $x$. However, for $u$ left of $x$, in the
evaluation game, the value of $u$ is already set when $x$ is to be
assigned. Thus already existing non-trivial partial strategies for $u$
cannot be combined with a branching decision, and so this resolution
step is blocked. However, if both the strategies are identical, or if
one of them is trivial (unspecified), then the non-trivial
strategy can be carried forward while maintaining the desired
invariant. Checking whether strategies are identical can itself be
hard, making verification of the proof difficult. In \MRes, this is
handled by choosing a particular representation called merge maps,
where isomorphism checks are easy.

Now we can describe the proof system itself. 
First we describe \emph{merge maps}.
\begin{definition}
	\emph{Merge maps} are  deterministic
	branching programs, specified by a sequence of instructions of one of
	the following two forms:
	\begin{itemize}
		\item 
		\texttt{$\langle \text{line~}\ell\rangle: b$}, where $b \in
		\{*,0,1\}$.\footnote{In \cite{DBLP:journals/jar/BeyersdorffBM21}, the notation used is
			$b \in \{*,u,\overline{u}\}$;
			$u,\overline{u},*$ denote  $u=1,u=0$, undefined respectively.
		}  \\ Merge maps containing a single such instruction are
		called simple. In particular, if $b=*$, then they are called
		trivial.
		\item
		\texttt{$\langle \text{line ~}\ell\rangle: \text{~If~} x=0 \text{~then
				go to~} \langle \text{line~} \ell_1\rangle \text{~else go to~}
			\langle \text{line~} \ell_2\rangle$}, for some $\ell_1, \ell_2 <
		\ell$. In a merge map $M$ for $u$, all queried variables $x$ must
		precede $u$ in the quantifier prefix. \\ Merge maps with such
		instructions are called complex.
	\end{itemize}
	(All line numbers are natural numbers.)  The merge map $M^u$ computes
	a partial strategy for the universal variable $u$ starting at the
	largest line number (the leading instruction) and following the
	instructions in the natural way. The value $*$ denotes an undefined
	value.
\end{definition}
\begin{definition}
	Two merge maps $M_1$ and $M_2$ are said to be \emph{consistent}, denoted $M_1 \bowtie M_2$, if for every line number $i$ appearing in both
	$M_1,M_2$, the instructions with line number $i$ are
	identical.
\end{definition}
When two merge maps, $M_1$ and $M_2$, are consistent, it is possible to build the merge map: \texttt{If $x = 0$ then go to $M_1$ else go to $M_2$} without repeating the common parts of $M_1$ and $M_2$. To be more precise, the new merge map will contain all instructions of $M_1$ and $M_2$ and the following additional instruction: \texttt{If $x = 0$ then go to $\langle$leading instruction of $M_1$$\rangle$ else go to $\langle$leading instruction of $M_2$$\rangle$}.

\begin{definition}
	Two merge maps $M_1, M_2$ are said to be isomorphic,
	denoted $M_1 \simeq M_2$, if there is a bijection between the line
	numbers in $M_1$ and $M_2$ that transforms $M_1$ to $M_2$ in
	the natural way. 
\end{definition}

For the remainder of this section let $\Phi=Q\cdot\phi$ be a QBF with existential variables $X$ and universal variables $U$.
\begin{definition}
	The \emph{proof system \MRes} has the following rules:
	\begin{enumerate}
		\item \emph{Axiom:} For a clause $A$ in the matrix $\phi$, let $C$ be the
		existential part of $A$. For each universal variable $u$, let
		$b_u$ be the value $u$ must take to falsify $A$; if
		$u\not\in\var(A)$, then $b_u=*$. For any natural number $i$, the
		line $(C,\{M^u:u\in U\})$ where each $M^u$ is the simple
		merge map $\langle i\rangle: b_u$ can be derived in \MRes.
		\item \emph{Resolution:} From lines $L_a=(C_a,\{M^u_a:u\in
		U\})$ for $a\in \{0,1\}$, in \MRes, the line
		$L=(C,\{M^u:u\in U\})$ can be derived, where for some
		$x\in X$,
		\begin{itemize}
			\item    $C = \res(C_0,C_1,x)$, and 
			\item for each $u\in U$; either
			$M^u_a$ is trivial and $M^u=M^u_{1-a}$ for some $a$; or
			$M^u=M^u_0\simeq M^u_1$; or
			$x$ precedes $u$, $M_1 \bowtie M_2$ and $M^u$ %
			has all the instructions of $M^u_1$ and $M^u_2$ in addition to the following instruction: \\
			\texttt{If $x = 0$ then go to $\langle$leading instruction of $M^u_1$$\rangle$ else go to $\langle$leading instruction of $M^u_2$$\rangle$}.\\
			The line number of this leading instruction is the  number (position) of the line $L$ in the derivation.
		\end{itemize}
		With slight abuse of notation, we will call $L$ the resolvent of $L_0$ and $L_1$ with pivot $x$, and denote this by $L = \res(L_0,L_1,x)$.
		
		Note that \cite{DBLP:journals/jar/BeyersdorffBM21} also requires that the positive literal of the pivot  appears in the first argument, so $x \in C_0$ (i.e.~the clause at line $L_0$) and $\overline{x} \in C_1$ (the clause at line $L_1$). However, this was only for syntactic convenience, and the way we formulate our arguments, this is not necessary.)
	\end{enumerate}
\end{definition}
Note that the entire merge maps are not stored at each line, only the
leading instruction specific to the line.  Due to consistency, this is
enough information to build the entire map from the derivation. As
noted in \cite{DBLP:journals/jar/BeyersdorffBM21} (Proposition 19), for lines within the same
derivation, the corresponding merge maps are always
consistent. Therefore, in the above definition, we don't have to
explicitly do a consistency check.
\begin{definition}
	A \emph{refutation} is a derivation using these rules and ending in a line
	with the empty existential clause. The size of the refutation is the
	number of lines.
\end{definition}

In the rest of this paper, we will denote refutations by the Greek letter $\Pi$. A refutation can be represented as a graph (with edges directed from the antecedents to the consequent, hence from the axioms to the final line). We denote the graph corresponding to refutation $\Pi$ by $G_{\Pi}$.
The lines of $\Pi$ will be denoted by $L$, $L_1$, $L_2$, $L'$, $L''$ etc. For lines $L$, $L_i$ and $L'$, and universal variable $z$, the respective clause, merge map and the function computed by the merge map will be denoted by $C$, $M^z$, $h^z$, $C_i$, $M^z_i$, $h^z_i$ and $C'$, $(M')^z$, $(h')^z$ respectively.

\begin{definition}
	Let  $Y$ be a subset of the existential variables $X$ of $\Phi$. We say that an $\MRes$
	refutation $\Pi$ of $\Phi$ is \emph{$Y$-regular} if for each $y\in Y$, there is no
	leaf-to-root path in $G_{\Pi}$ that uses $y$ as pivot more
	than once. An $X$-regular proof is simply called a \emph{regular proof}.
	If $G_\Pi$ is a tree, then we say that $\Pi$ is a \emph{tree-like proof}.
\end{definition}

\begin{example}\label{ex:Equality}
	We reproduce from \cite{DBLP:journals/jar/BeyersdorffBM21} a small example to illustrate how $\MRes$ operates. The formulas to be refuted are the Equality formulas from \cite{BBH-LMCS19}, defined as follows:
	The \emph{Equality family} is the QBF family whose $n$th instance has
	the prefix $\exists x_1, \dots, x_n, \forall u_1, \dots, u_n,
	\exists t_1, \dots ,t_n$ and the matrix consisting of the clauses
	$\{x_i,u_i,t_i\},\{\overline{x}_i,\overline{u}_i,t_i\}$ for $i \in [n]$, and
	$\{\overline{t}_1, \dots ,\overline{t}_n\}$.
	
	In \cite{DBLP:journals/jar/BeyersdorffBM21} (Example~3), these formulas are shown to have linear-size refutations in the system $\rLDQRes$ denoting reductionless $\LDQRes$, the fragment of $\LDQRes$ where all reductions are postponed to the very end (no reduction step is followed by a resolution step). Later in \cite{DBLP:journals/jar/BeyersdorffBM21} (Theorem~22), $\MRes$ is
	shown to simulate reductionless $\LDQRes$. Hence these formulas are easy to refute in $\MRes$.
On the other hand, these formulas are known to require exponential-size refutations in $\QRes$, $\QURes$, $\QCP$ \cite{BBH-LMCS19}, $\ExpRes$ and $\IR$ \cite{BB-TOCS20} (cf.\ \cite{BeyersdorffB17-ECCC} on how to apply the lower bound technique from \cite{BB-TOCS20} to the Equality formulas).

        Here, we directly present the
	implied linear-size $\MRes$ refutations (in fact, these refutations are also tree-like and regular) for the Equality formulas. 
	
	First, we download the axioms. Line $0$ downloads the long clause $\{\overline{t}_1, \dots ,\overline{t}_n\}$, with
	all trivial merge maps. The next $2n$ lines download the short axiom
	clauses. Letting $i\in [n]$, we define these lines as follows: 
	Line $2i-1$ is the clause $\{x_i,t_i\}$ with merge map $0$ for $u_i$
	and all other merge maps are trivial.
	Line $2i$ is the clause $\{\overline{x}_i,t_i\}$ with merge map $1$ for $u_i$
	and all other merge maps are trivial.
	
	For $i\in[n]$, line $2n+i$ is obtained by applying the merge
	resolution rule on lines $2i-1$ and $2i$. This gives the clause
	$\{t_i\}$; the merge maps for $j\neq i$ are trivial, and the merge map
	for $u_i$ has the instruction:
	\texttt{$\text{~If~} x_i=0 \text{~then
		go to~} \langle \text{line~} 2i-1\rangle \text{~else go to~}
	\langle \text{line~} 2i \rangle$}. 
	
	At line $3n+1$, applying merge resolution on lines 0 and $2n+1$, we
	obtain the clause $\{\overline{t}_2, \dots ,\overline{t}_n\}$. The merge map for
	$u_1$ is taken from line $2n+1$, since at line 0 it is trivial.
	
	Now for $i\in [2,n]$, line $3n+i$ is obtained by applying merge
	resolution on lines $2n+i$ and $3n+i-1$. This gives the clause
	$\{\overline{t}_{i+1}, \dots ,\overline{t}_n\}$. The merge map for $u_i$ is
	taken from line $2n+i$ since at line $3n+i-1$ it is trivial.  For
	$j< i$, the merge map for $u_j$ is taken from line $3n+i-1$ since at
	line $2n+i$ it is trivial. Effectively, at this line, for all $j\le
	i$, the merge map for $u_j$ is from line $2n+j$, and for all $j >i$, the merge map for $u_j$ is trivial.
	
	Line $4n$ derives the empty clause and the strategy computing, for each
	$i\in[n]$, $u_i=x_i$. This completes the refutation.
\end{example}

A  crucial fact about the proof system $\MRes$,  shown in \cite{DBLP:journals/jar/BeyersdorffBM21}, is that the merge maps at the final line of an $\MRes$ refutation compute a countermodel for the QBF. To 
establish this fact, some stronger properties of $\MRes$ refutations are
established and will be useful to us. We restate the relevant
properties here.

\begin{lemma}[Extracted/adapted from \cite{DBLP:journals/jar/BeyersdorffBM21} Section 4.3, (Proof of Lemma 21)] \label{ref:soundness}
Let $\Phi=Q\cdot\phi$ be a QBF with existential variables $X$ and universal variables $U$. 
Let $\Pi \defeq L_1, \ldots, L_m$ be an $\MRes$ refutation of $\Phi$, where each $L_i = (C_i , \{M_i^u \mid u \in U \})$. Further, for each $i \in [m]$,
\begin{itemize}
\item let $\alpha_i$ be the minimal partial assignment falsifying $C_i$,
\item let $A_i$ be the set of assignments to $X$ consistent with $\alpha_i$,
\item for each $u \in U$, let $h_i^u$ be the function computed by $M_i^u$,
\item for each $\alpha \in A_i$, let $h_i (\alpha)$ be the partial assignment which sets variable $u$ to $h_i^u(\alpha \restriction_{L_Q(u)})$ if $h_i^u(\alpha \restriction_{L_Q(u)}) \neq *$, and leaves it unset otherwise.
\end{itemize}
Then for each $\alpha \in A_i$, the (partial) assignment $(\alpha, h_i(\alpha))$ falsifies at least one clause of $\phi$ used in the sub-derivation of $L_i$.
\end{lemma}

\begin{proposition}[\cite{DBLP:journals/jar/BeyersdorffBM21}] \label{prop:proof-mergemap-isomorphism}
Let $\Phi$ and $\Pi$ be as defined in \Cref{ref:soundness}. Then, for all $u \in U$, $M^u_m$ is isomorphic to a subgraph of $G_{\Pi}$ {\upshape(}up to path contraction{\upshape)}.
\end{proposition}

\section{Lower bounds}
\label{sec:lower bounds}

We now start to explore lower bounds for \MRes. %
In \Cref{subsec:generic-lb}, we show how to construct generic hard formulas for tree-like and regular $\MRes$.
In \Cref{subsec:QParity-lb,subsec:LQParity-lb,subsec:CR-lb-tree-and-regular,subsec:KBKFlq-lb}, we prove lower bounds for specific QBF formulas.

\subsection{Lower bounds for generic formulas}
\label{subsec:generic-lb}

The following theorem, implicit in \cite{DBLP:journals/jar/BeyersdorffBM21}, is an immediate consequence of \Cref{ref:soundness} and \Cref{prop:proof-mergemap-isomorphism}.
\begin{theorem}
	\label{thm:lb-from-strategy-extraction}
	Let $\Phi=Q\cdot\phi$ be a false QBF with existential variables $X$ and universal variables $U$. If, for every countermodel of $\Phi$, the function for some universal variable $u$ requires size at least $s$ to compute by branching programs (resp.~decision trees, read-once branching programs), then every $\MRes$ (resp.~tree-like $\MRes$, regular $\MRes$) refutation of $\Phi$ has  size at least $s$.
\end{theorem}

Currently, no exponential lower bounds are known for general branching programs. Therefore, we cannot use the above theorem to prove lower bounds for general $\MRes$. However, we can use it to prove exponential lower bounds for tree-like and regular $\MRes$. To do so, 
we need a QBF whose countermodel requires exponential decision-trees (resp.~read-once branching programs).
We now show how to construct such QBFs.
This follows the method used, for instance,  in \cite[Sec.~4.1]{BCJ-ToCT-Sep2019} and \cite[Sec.~6]{PeitlSS19a}. 

Let $f \colon X \to \{0,1\}$ be a Boolean function, let $C_f$ be a Boolean circuit encoding $f$, and let $u$ be a variable not in $X$. Using the Tseitin transformation \cite{Tseitin83}, we can construct a CNF formula $\phi(X,u,Y)$ such that $\exists Y. \phi(X,u,Y)$ is logically equivalent to $C_f(X) \neq u$. Then, the QBF formula $\Phi \coloneqq \exists X \forall u \exists Y. \phi(X,u,Y)$, called the QBF encoding of $f$, is a false QBF formula with $f$ as the unique winning strategy. Moreover, the size of $\Phi$ is polynomial in the size of $C_f$. This is the desired QBF formula.

\subsection{The QParity formulas}
\label{subsec:QParity-lb}
We now turn our attention to lower bounds for specific formulas. We start with the $\QParity$ formulas in this \namecref{subsec:QParity-lb}.
These are the formulas obtained by the Tseitin transformation described above, using a linear-size read-once branching program computing the parity function. 
These formulas were defined in \cite{BCJ-ToCT-Sep2019} where they were shown to be hard for $\QRes$ and $\QURes$. It was also shown that these formulas are easy for the expansion-based systems $\ExpRes$, $\IR$ and $\IRM$. It was hence concluded that $\QRes$ and $\QURes$ do not simulate $\ExpRes$, $\IR$ and $\IRM$.

\label{subsec:qparity-bp-lifting}

Before we define the formulas, we set up some notation.
For variables $o,o_1,o_2$, let $\xor(o_1,o)$ and $\xor(o_1,o_2,o)$ be the following sets of clauses:
\begin{IEEEeqnarray*}{rCl}
	\xor(o_1,o) &=& \{\overline{o_1} \vee o, o_1 \vee \overline{o} \},\\
	\xor(o_1,o_2,o) &=& \{\overline{o_1} \vee \overline{o_2} \vee \overline{o}, \overline{o_1} \vee o_2 \vee o, o_1 \vee \overline{o_2} \vee o, o_1 \vee o_2 \vee \overline{o} \}
\end{IEEEeqnarray*}
Note that $\xor$ on a set of variables is just the CNF representation of the constraint that the number of variables set to true is even. 
That is, $\xor(o_1,o)$ is satisfied iff $o \equiv o_1 \pmod{2}$, and $\xor(o_1,o_2,o)$ is satisfied iff $o \equiv o_1 + o_2 \pmod{2}$.

\begin{definition}
	The $\QParity_n$ formula \cite{BCJ-ToCT-Sep2019} is the QBF $\exists x_1, \ldots , x_n, \forall z, \exists t_1, \ldots , t_n.\ \big(\bigwedge_{i \in [n+1]} \phi_n^{i}\big)$ where
	\begin{IEEEeqnarray*}{cCl}
		\phi_n^{1} &=& \xor(x_1,t_1);\\
		\phi_n^{i} &=& \xor(t_{i-1},x_i,t_i), \qquad \forall i \in [2,n];\\
		\phi_n^{n+1} &=& \{t_n \vee z, \overline{t_n} \vee \overline{z}\}.
	\end{IEEEeqnarray*}
\end{definition}

The QBFs are false: they claim that there exist $x_1,\dots, x_n$ such that $x_1 + \dots + x_n$ is neither congruent to $0$ nor $1$ modulo $2$. Note that the only winning strategy for the universal player is to play $z$ satisfying $z \equiv x_1 + \dots + x_n \pmod{2}$.
 
\begin{theorem}
  \label{thm:treeMRes-parity-lb}
Every tree-like $\MRes$ refutation of $\QParity_n$ has size at least $2^n$.
\end{theorem}
\begin{proof}
	It is a folklore fact that the $n$-input parity function %
	requires decision-trees of size at least $2^n$.
	From \Cref{thm:lb-from-strategy-extraction}, we obtain the desired lower bound.
\end{proof}

\subsection{The LQParity formulas}
\label{subsec:LQParity-lb}
We now turn our attention to the $\LQParity$ formulas.
These formulas are variants of the $\QParity$ formulas, and were originally defined in \cite{BCJ-ToCT-Sep2019}.  The variant is designed 
to allow the $\QRes$ lower bound arguments for $\QParity$ to be adapted also to $\LDQRes$. Like the $\QParity$ formulas, these formulas are easy for several proof systems, including $\ExpRes$, $\IR$ and $\IRM$, but were shown to be hard for$\QURes$ and  $\LDQRes$. This then established that $\LDQRes$ does not simulate $\ExpRes$, $\IR$ and $\IRM$ \cite{BCJ-ToCT-Sep2019}. %

We now describe the formulas. They are obtained from the $\QParity$ formulas by duplicating each clause except those in $\phi_n^{n+1}$, and inserting the universal variable $z$ in one copy and its negation $\overline{z}$ in the other. 
Formally, they can be defined as follows: 
For variables $o,o_1,o_2,z$, let $\xorl(o_1,o,z)$ and $\xorl(o_1,o_2,o,z)$ be the following sets of clauses:
\begin{IEEEeqnarray*}{rCl}
	\xorl(o_1,o,z) &=& \{\overline{o_1} \vee o \vee z, o_1 \vee \overline{o} \vee z \},\\
	\xorl(o_1,o_2,o,z) &=& \{\overline{o_1} \vee \overline{o_2} \vee \overline{o} \vee z, \overline{o_1} \vee o_2 \vee o \vee z, o_1 \vee \overline{o_2} \vee o \vee z, o_1 \vee o_2 \vee \overline{o} \vee z\}
\end{IEEEeqnarray*}

\begin{definition}
	The $\LQParity_n$ formula \cite{BCJ-ToCT-Sep2019} is the QBF $\exists x_1, \ldots , x_n, \forall z, \exists t_1, \ldots , t_n.\ \big( \bigwedge_{i \in [n+1]} \phi_n^{i} \big)$ where
	\begin{IEEEeqnarray*}{cCl}
		\phi_n^{1} &=& \xorl(x_1,t_1,z) \cup \xorl(x_1,t_1,\overline{z}),\\
		\phi_n^{i} &=& \xorl(t_{i-1},x_i,t_i,z) \cup \xorl(t_{i-1},x_i,t_i,\overline{z}) \qquad \forall i \in [2,n],\\
		\phi_n^{n+1} &=& \{t_n \vee z,
		\overline{t_n} \vee \overline{z}\}.
	\end{IEEEeqnarray*}
\end{definition}
For $i,j \in [n+1], i \le j$, let $\phi_n^{[i,j]}$ denote
$\bigwedge_{k \in [i,j]} \phi_n^{k}$. Also, let $X=\{x_1,\ldots,x_n\}$ and
$T=\{t_1,\ldots,t_n\}$.

\begin{observation}
	\label{obs:var-in-phi}
	\begin{enumerate*}[label=(\alph*)]
		\item For each $i\in [n]$, and each $C\in \phi_n^i$, $\{x_i,t_i\}
		\subseteq \var(C)$; and
		\item for each $i\in [n+1]\setminus\{1\}$, and each $C\in
		\phi_n^i$, $\{t_{i-1}\} \subseteq \var(C)$.
	\end{enumerate*}
	
\end{observation}

We will now show that $\LQParity$ formulas require exponential-size refutations in regular $\MRes$.

\begin{restatable}{theorem}{restateLQParityLB}
	\label{thm:LQParity-lb}
	Every $T$-regular refutation of $\LQParity_n$ in $\MRes$, and hence any regular $\MRes$ refutation,  has size at least	$2^n$.
\end{restatable}

The proof proceeds as follows: Let $\Pi$ be a $T$-regular \MRes\
refutation of $\LQParity_n$. Since every axiom has a variable from $T$
while the final clause in $\Pi$ is empty, there is a maximal
``component'' (say $\mathcal{S}$) of the proof leading to and including the final line,
where all clauses are $T$-free.  The clauses in this component involve
only the $X$ variables. We show that the ``boundary'' ($\partial{\mathcal{S}}$) of this component
is large, by showing in \Cref{lem:LQParity-width} that each clause
at the boundary must be wide.  (This idea was used in \cite{PeitlSS19a} to show
that $\CR$ is hard for reductionless $\LDQRes$.) To establish the
width bound, we note that no lines have trivial strategies.  Since the
pivots at the boundary are variables from $T$, the merge maps incoming
into each boundary resolution must be isomorphic. By carefully
analysing which axiom clauses can and must be used to derive lines just
above the boundary (\Cref{lem:LQParity-intervals}), we conclude
that the merge maps must be simple, yielding the lower bound.
To fill in all the details,
we first describe some properties (\Cref{lem:properties_of_derived_clauses}) of $\Pi$  that will be used in obtaining this result.

Recall that the lines of $\Pi$ have mergemaps for the  universal variable.
Since $\LQParity$ formulas have a single universal variable, we avoid the superscript $z$. Thus a line $L$ (resp.\ $L_1, L_2, L', L''$ etc) has  a merge map $M$ (resp.\ $M_1,M_2, M', M''$, etc) for $z$ and computes the function $h$
(resp,\ $h_1,h_2,h',h''$, etc.).

Let $G_{\Pi}$ be the derivation graph corresponding to $\Pi$ (with edges directed from the antecedents to the consequent, hence from the axioms to the final line). We will refer to the nodes of this graph by the corresponding line. For $L, L' \in \Pi$, we will say $L \leadsto L'$ if there is a path from $L$ to $L'$ in $G_{\Pi}$. 

For a line $L \in \Pi$,
let $\Pi_{L}$ be the minimal sub-derivation of $L$, and let $G_{\Pi_{L}}$ be the corresponding subgraph of $G_{\Pi}$ with sink $L$. Define $\UsedConstraints(\Pi_{L}) = \{\phi_n^i \mid i \in [n+1], \leaves(G_{\Pi_{L}}) \cap \phi_n^{i} \neq \emptyset \}$, and  $\UsedConstraintInd(\Pi_{L}) = \{ i \in [n+1] \mid \phi_n^{i} \in \UsedConstraints(\Pi_{L})\}$. (The abbreviations $\UsedConstraints$ and $\UsedConstraintInd$ stand for UsedConstraints and UsedConstraintsIndex respectively.) Note that for any leaf  $L$,  $\UsedConstraintInd(\Pi_L)$ is a singleton.

Define $\mathcal{S}$ to be the set of those lines in $\Pi$ where the
clause part has no $T$ variable and furthermore there is a path in
$G_\Pi$ from the line to the final empty clause via lines where all
the clauses also have no $T$ variables. Let $\partial{\mathcal{S}}$, called the boundary of $\mathcal{S}$, denote the
set of leaves in the subgraph of $G_\Pi$ restricted to $\mathcal{S}$;
these are lines that are in $\mathcal{S}$ but their parents are not
in $\mathcal{S}$.
Note that no leaf of
$\Pi$ is in $\mathcal{S}$ because all leaves of $G_\Pi$ contain a
variable in $T$.

\begin{lemma} \label{lem:properties_of_derived_clauses}
Let $L=(C,M)$ be a line of $~\Pi$. Then
$\UsedConstraintInd(\Pi_{L})$ is an
  interval $[i,j]$ for some $1 \le i \le j \le n+1$. Furthermore, 
  {\upshape(}below $i,j$ refer to the endpoints of this interval\,{\upshape)}
  \begin{enumerate}
\item \label{properties:in-interval}
  For all $k \in [i,j-1]$,  $t_k \not\in \var(C)$.
\item \label{properties:low-end}
  If $i> 1$, then $t_{i-1} \in \var(C)$. 
\item \label{properties:high-end}
  If  $j \le n$, then $t_j \in \var(C)$. 
\item \label{properties:singleton}
  $|\var(C)\cap T| = 1$ iff $[i,j]$ contains exactly one of $1,n+1$.\\
  $\var(C)\cap T = \emptyset$ iff $[i,j] = [1,n+1]$.
\item \label{properties:X}
  For all $k \in [i,j] \cap[1,n]$,  $x_{k} \in \var(C) \cup \var(M)$.
\end{enumerate}
\end{lemma}
\begin{proof}
	Let $I=\UsedConstraintInd(\Pi_L)$. 
	Assume, to the contrary, that $I$ is not an interval; for some $k\in
	[2,n]$, $I$ contains an index $i < k$ and an index $j > k$, but does
	not contain $k$.  Let $L'$ be the first line in $\Pi$ such that
	$\UsedConstraintInd(\Pi_{L'})$ intersects both $[1,k-1]$ and
	$[k+1,n+1]$. Since leaves have singleton $\UsedConstraintInd$ sets,
	$L'$ is not a leaf. Say $L'=\res(L'',L''',v)$. Assume that
	$\UsedConstraintInd(\Pi_{L''})\subseteq[1,k-1]$ and
	$\UsedConstraintInd(\Pi_{L'''})\subseteq[k+1,n+1]$; the argument for
	the other case is identical.  So $v \in
	\var_{\exists}(\UsedConstraints(\Pi_{L''}))\subseteq
	\var_{\exists}(\phi_n^{[1,k-1]})$, and $v \in
	\var_{\exists}(\UsedConstraints(\Pi_{L'''}))\subseteq
	\var_{\exists}(\phi_n^{[k+1,n+1]})$. But
	$\var_{\exists}(\phi_n^{[1,k-1]})$ and
	$\var_{\exists}(\phi_n^{[k+1,n+1]})$ are disjoint, a contradiction.
	
	Fixing $i,j$ so that $I=\UsedConstraintInd(\Pi_L) = [i,j]$, we now prove the remaining statements in the Lemma.
	\begin{enumerate}
		\item Fix any $k\in[i,j-1]$. Note that $\{k,k+1\} \subseteq
		\UsedConstraintInd(\Pi_L)$. Let $L'$ be the first line in $\Pi_{L}$
		such that $\{k,k+1\} \subseteq \UsedConstraintInd(\Pi_{L'})$. Say
		$L'$ is obtained as $\res(L'',L''',v)$. Assume that
		$\UsedConstraintInd(\Pi_{L''})$ contributes $k$ and
		$\UsedConstraintInd(\Pi_{L''})$ contributes $k+1$; the other case is
		symmetric.  Since $\UsedConstraintInd(\Pi_{L''})$ must also be an
		interval, and since it contains $k$ but not $k+1$,
		$\UsedConstraintInd(\Pi_{L''})\subseteq [1,k]\cap
		\UsedConstraintInd(\Pi_{L}) = [i,k]$. Similarly,
		$\UsedConstraintInd(\Pi_{L'''})\subseteq [k+1,j]$. The pivot variable $v$
		must thus belong to both $\phi_n^{[i,k]}$ and $\phi_n^{[k+1,j]}$;
		the only such existential variable is $t_k$. Hence each $t_k$ is
		used as a pivot in $\Pi_L$.
		
		Since $\Pi$ is $T$-regular, and since $t_k$ is used as a pivot to
		derive $L'$ inside $\Pi_L$, it cannot reappear in any line on any
		path from (including) $L'$ to the final clause. Hence it does not
		appear in $L$.

		\item Let $i > 1$. By \Cref{obs:var-in-phi}, $t_{i-1}$
		appears in at least one axiom used in $\Pi_L$. Assume to the
		contrary that $t_{i-1} \not\in \var(C)$. Let $\rho_C$ be the minimal
		partial assignment falsifying $C$. By assumption, $\rho_C$ does not
		set $t_{i-1}$, and by item~(\ref{properties:in-interval}) above,
		$\rho_C$ does not set any variable $t_k$ with $i \le k < j$. Extend
		$\rho_C$ arbitrarily to all unassigned variables in $(X \cup T)
		\setminus \{t_{i-1}, \ldots, t_{j-1} \}$ to get $\rho_1$. Since the
		merge map $M$ does not depend on variables in $T$, the partial
		assignment $\rho_1$ is sufficient to evaluate $M$ and $h$. Define
		the value $y$ as follows:
		\[y = \left\{
		\begin{array}{ll}
			\rho_1(t_j) & \text{ if } j \le n, \\
			h(\rho_1) & \text{ if } j = n+1.
		\end{array}
		\right.
		\]
		For $b\in \{0,1\}$, let $\rho_1^b$ denote the extension of $\rho_1$
		by $t_{i-1}=b$.  Exactly one of $\rho_1^0, \rho_1^1$ satisfies the
		equation $t_{i-1} + x_i + x_{i+1} + \ldots + x_j + y \equiv 0 \bmod
		2$; let this extension be $\rho_2$. Then there is a unique extension
		$\alpha$ of $\rho_2$ to $X\cup T$ such that
		\begin{itemize} 
			\item if $j \le n$, then $\alpha$ satisfies the existential part of
			all clauses in $\phi_n^{[i,j]}$;
			\item if $j=n+1$, then $(\alpha, h(\rho_1))$ satisfies all clauses
			in $\phi_n^{[i,j]}$. (That is, assigning $X\cup T$ according to
			$\alpha$ and assigning $z$ the value $h(\rho_1)$ satisfies
			$\phi_n^{[i,j]}$.)
		\end{itemize}
		(To find $\alpha$, work backwards from $y$ to determine the
		appropriate values of $t_{j-1}, t_{j-2}, \ldots , t_{i}$ to
		satisfy $\phi_n^j$, $\phi_n^{j-1}$, $\ldots$, $\phi_n^i$.)
		
		Note that $h(\rho_1) = h(\rho_2) = h(\alpha)$. So $(\alpha,
		h(\alpha))$ falsifies $C$ (since it extends $\rho_C$) and satisfies
		all axiom clauses used to derive $L$. This contradicts
		\Cref{ref:soundness}.
		
		\item Let $j \le n$. Assume to the contrary that
		$t_j\not\in\var(C)$. The argument is identical to that in
		item~\ref{properties:low-end} (only the indices differ): $\rho_C$
		falsifies $C$; $\rho_1$ extends it arbitrarily to all unassigned
		variables in $(X \cup T) \setminus \{t_{i}, \ldots, t_{j} \}$;
		$\rho_2$ is the extension of $\rho_1$ obtained by setting $t_j$ so
		as to satisfy the equation $t_{i-1} + x_i + x_{i+1} + \ldots + x_j +
		t_j \equiv 0 \bmod 2$; (Here, if $i=1$, discard $t_0$ from the
		equation; i.e.\ assume $t_0=0$); $\alpha$ is the unique extension of
		$\rho_2$ to $X\cup T$ satisfying $\phi_n^{[i,j]}$ (To obtain
		$\alpha$, work forwards obtaining $t_i, t_{i+1}, \ldots,
		t_{j-1}$). Now $(\alpha, h(\alpha))$ contradicts
		\Cref{ref:soundness}.

		\item Since $\UsedConstraintInd(\Pi_L) = [i,j]$, variables $t_k$ for
		$k \not\in[i-1,j]$ do not appear in any of the used axioms
		(\Cref{obs:var-in-phi}) and hence do not appear in $C$.
		By the preceding three items, $\var(C)\cap T$ does not include any
		$t_k$ with $k\in [i,j-1]$, includes $t_{i-1}$ whenever $i>1$, and
		includes $t_j$ whenever $j<n+1$. The claim follows.

		\item
		Assume to the contrary that for some $k\in[i,j]$,
		$x_k\not\in\var(C)\cup\var(M)$. The argument is similar to that in
		item~(\ref{properties:low-end}): $\rho_C$ falsifies $C$; $\rho_1$
		extends it arbitrarily to all unassigned variables in
		$(X\setminus\{x_k\}) \cup (T \setminus \{t_{i}, \ldots, t_{j-1}
		\})$; $y$ is the value of $t_j$ if $j\le n$ and the value of $h$
		otherwise (since $x_k\not\in\var(M)$, $\rho_1$ is sufficient to
		evaluate $h$); $\rho_2$ is the extension of $\rho_1$ obtained by
		setting $x_k$ so as to satisfy the equation $t_{i-1} + x_i + x_{i+1}
		+ \ldots + x_j + y \equiv 0 \bmod 2$; (Here, if $i=1$, discard $t_0$
		from the equation; i.e.\ assume $t_0=0$); $\alpha$ is the unique
		extension of $\rho_2$ to $X\cup T$ satisfying $\phi_n^{[i,j]}$ (To
		obtain $\alpha$, work forwards from $t_i$ towards $t_{j-1}$). Now
		$(\alpha, h(\alpha))$ contradicts \Cref{ref:soundness}.
		\qedhere
	\end{enumerate}
\end{proof}

\begin{lemma}\label{lem:LQParity-intervals}
  Let $L \in \partial{\mathcal{S}}$ be derived in $\Pi$ as $L = \res(L',L'',t_k)$. Then
  $\UsedConstraintInd(\Pi_L) = [1,n+1]$, and
  $\UsedConstraintInd(\Pi_{L'}),\UsedConstraintInd(\Pi_{L''})$
  partition $[1,n+1]$ into $[1,k],[k+1,n+1]$.
\end{lemma}
\begin{proof}
	Since $L\in \partial{\mathcal{S}}$, $L$ has no variable from $T$. By
	\Cref{lem:properties_of_derived_clauses}(\ref{properties:singleton}),
	$\UsedConstraintInd(\Pi_L) = [1,n+1]$.	
	
	Since $L=\res(L',L'',t_k)$, we see that $\var(C')\cap T = \var(C'')\cap T =
	\{t_k\}$. By
	\Cref{lem:properties_of_derived_clauses}(\ref{properties:low-end},\ref{properties:high-end},\ref{properties:singleton}), we see that
	$\UsedConstraintInd(\Pi_{L'}), \UsedConstraintInd(\Pi_{L''})
	\in \{[1,k], [k+1,n+1] \}$.	
	If both $\UsedConstraintInd(\Pi_{L'}), \UsedConstraintInd(\Pi_{L''})$ equal
	$[k+1,n+1]$, then $\UsedConstraintInd(\Pi_L)=[k+1,n+1]$, contradicting
	$\UsedConstraintInd(\Pi_l) = [1,n+1]$.
	On the other hand, if both $\UsedConstraintInd(\Pi_{L'}), \UsedConstraintInd(\Pi_{L''})$
	equal $[1,k]$, then $\UsedConstraintInd(\Pi_{L})=[1,k]$. Since $t_k$
	is a pivot variable, $k \le n$, contradicting   $\UsedConstraintInd(\Pi_L) = [1,n+1]$.	
	Hence one each of $\UsedConstraintInd(\Pi_{L'}), \UsedConstraintInd(\Pi_{L''})$
	equals $[1,k]$ and $[k+1,n+1]$ as claimed.
\end{proof}

\begin{lemma}\label{lem:LQParity-width}
For all $L \in \partial{\mathcal{S}}$, $\width(C) = n$.
\end{lemma}
\begin{proof}
	Let $L \in \partial{\mathcal{S}}$ be derived in $\Pi$ as $L = \res(L',L'',t_k)$.  Since
	all axioms create non-trivial strategies, neither $M'$ nor $M''$
	equals $*$. By the rules of MRes, $M'=M''=M\neq *$. We will show
	that in fact $M$ must be a constant strategy, $M \in \{0,1\}$.
	
	By definition of $\partial{\mathcal{S}}$, $\var(C)\cap T = \emptyset$, and hence
	$\var(C')\cap T = \var(C'') \cap T= \{t_k\}$.  By
	\Cref{lem:LQParity-intervals}, $\UsedConstraintInd(\Pi_L) =
	[1,n+1]$ is partitioned by $\UsedConstraintInd(\Pi_{L'})$ and
	$\UsedConstraintInd(\Pi_{L''})$ into $[1,k], [k+1,n+1]$.
	
	Assume $\UsedConstraintInd(\Pi_{L'}) = [1,k]$,
	$\UsedConstraintInd(\Pi_{L''}) = [k+1,n+1]$; the argument in the
	other case is identical.  Then $\var(M) = \var(M') \subseteq
	\var(\phi^{[1,k]}) \cap X = \{x_1,\ldots,x_k\}$, and $\var(M) =
	\var(M'') \subseteq \var(\phi^{[k+1,n+1]}) \cap X =
	\{x_{k+1},\ldots,x_n\}$. The only way both these conditions can be
	satisfied is if $\var(M)=\emptyset$; that is, $M$ is a constant
	strategy.
	
	Since $\UsedConstraintInd(\Pi_{L}) = [1,n+1]$ and
	$\var(M)=\emptyset$, 
	\Cref{lem:properties_of_derived_clauses}(\ref{properties:X}) implies that $X \subseteq \var(C)$.  Therefore $\width(C) = n$.
\end{proof}

Now we can put together the proof of \Cref{thm:LQParity-lb}.
\begin{proof}[Proof of \Cref{thm:LQParity-lb}]
  Let $\Pi$ be a $T$-regular refutation of $\LQParity_n$ in $\MRes$. 
  Let $\mathcal{S}$ and $\partial{\mathcal{S}}$  be as defined in the beginning of this sub-section. By definition, 
  for each $L=(C,M) \in \mathcal{S}$, $\var(C) \subseteq X$. 
  Let $\widehat{\Pi} = \{ C \mid L=(C,M) \in \mathcal{S} \}$. Then $\widehat{\Pi}$ contains a propositional resolution refutation of $\mathcal{C} = \{ C \mid L=(C,M) \in \partial{\mathcal{S}}\}$. Therefore $\mathcal{C}$ is an unsatisfiable CNF formula over the $n$ variables in $X$. By \Cref{lem:LQParity-width}, each clause in $\mathcal{C}$ has width $n$ and so is falsified by exactly one assignment. Therefore, to ensure that each of the $2^n$ assignments falsifies some clause, (at least) $2^{n}$ clauses are required. Therefore $\card{\mathcal{C}} \geqslant 2^n$. Hence $\card{\Pi} \geqslant 2^n$.
\end{proof}

\subsection{The Completion Principle formulas}
\label{subsec:CR-lb-tree-and-regular}
We now move to the
Completion Principle ($\CR_{n}$) formulas 
introduced in \cite{Janota-Expansion-vs-QRes-TCS15}. %
From a proof-complexity viewpoint, these formulas are very simple: they have  polynomial-size, in fact linear-size, refutations in $\QRes$, and hence in $\QURes$, $\QCP$, $\ExpRes$ and $\IR$ \cite{Janota-QRes-and-CDCL-SAT16,Janota-Expansion-vs-QRes-TCS15}. The $\QRes$ refutations are even tree-like; \cite{MS16a}. They are known to be hard for $\QRes$ if the resolution pivots must respect the quantifier ordering (level-ordered $\QRes$); \cite{Janota-QRes-and-CDCL-SAT16,Janota-Expansion-vs-QRes-TCS15}. 

In this section, we prove that $\CR_{n}$ requires exponential-size proofs in tree-like and regular $\MRes$. Recall that no simulation is known between tree-like and regular $\MRes$, so these statements require separate proofs. We believe that $\CR_n$ requires exponential size refutations in general $\MRes$ as well, but we have not been able to prove this.

We first define the formulas:
\begin{restatable}{definition}{restateCRdef}
	The Completion Principle formulas $\CR_n$ \cite{Janota-Expansion-vs-QRes-TCS15} are defined as follows: 
	\begin{equation*}
		\CR_{n} = \displaystyle \mathop{\mathlarger{\exists}}_{i,j \in [n]} x_{ij}, \mathop{\mathlarger{\forall}} z, \mathop{\mathlarger{\exists}}_{i \in [n]} a_i, \mathop{\mathlarger{\exists}}_{j \in [n]} b_j.\ \left( \mathop{\mathlarger{\wedge}}_{i,j \in [n]} \left( A_{ij} \wedge B_{ij} \right) \right) \wedge L_A \wedge L_B
	\end{equation*}
	where
	$A_{ij} = x_{ij} \vee z \vee a_i$, 
	$B_{ij} = \overline{x_{ij}} \vee \overline{z} \vee b_j$,
	$L_A = \overline{a_1} \vee \cdots \vee \overline{a_n}$, and  
	$L_B = \overline{b_1} \vee \cdots \vee \overline{b_n}$.
	
	Let $X,A,B$ denote the variable sets $\{x_{ij}: i,j\in [n]\}$, $\{a_i:
	i\in [n]\}$, and $\{b_j: j\in [n]\}$.  It is convenient to think of
	the $X$ variables as arranged in an $n\times n$ matrix.
\end{restatable}

Intuitively, the formulas describe a completion game, played on the matrix
$$
\left(
\begin{array}{ccccccc}
	a_1 & \dots & a_1 & \dots & a_n & \dots & a_n\\
	b_1 & \dots & b_n& \dots & b_1 & \dots & b_n\\
\end{array}
\right)
$$
where the $\exists$-player first deletes exactly one cell per column and the $\forall$-player then chooses one row. The 
$\forall$-player wins if his row contains all of $A$ or all of $B$ (cf.\ \cite{Janota-Expansion-vs-QRes-TCS15}). In the formalisation, an assignment to the variable $x_{ij}$ indicates the choices of the first player in the column containing $a_i$ and $b_j$. 

\subsubsection{Lower bound for tree-like $\MRes$}
For the QBF $\CR_n$, the winning strategy for the universal
player (countermodel) is not unique. However, we show that all countermodels require large decision trees. 

\begin{lemma}
	\label{lem:treeMRes-CR}
	Every countermodel for $\CR_n$ has decision tree size complexity at least $2^n$. 
\end{lemma}
\begin{proof}
	We prove the size bound by showing that in every decision tree for every countermodel, all root-to-leaf paths query at least $n$ variables, and hence the decision tree has at least $2^n$ nodes.
	
	Assume to the contrary that some countermodel $h$ is computed by a decision tree $M$ that has  a root-to-leaf path $p$ querying less than $n$ variables. Then there exist $k,\ell \in [n]$ such that no variable from Row $k$ and no variable from Column $\ell$ is on this path.
	Let $\rho_p$ be the minimal partial assignment that takes this path in $M$, and let $\rho'$ be an arbitrary extension of $\rho_p$ to variables in $\{x_{ij} \mid i \neq k, j \neq \ell \}$. Consider the following extension of $\rho'$ to variables in $\left(X \setminus \{x_{k\ell}\}\right) \cup T$, giving assignment $\sigma$: \\
	Set all variables in row $k$ (other than $x_{k,\ell}$) to 1.\\
	Set all variables in column $\ell$ (other than $x_{k,\ell}$) to 0.\\
	Set $a_k$ and $b_\ell$ to 0 and all other $a_i,b_j$ variables to 1. 
	
	For $n\ge 2$, $\sigma$ satisfies all the clauses of
	$\CR_n$ except $A_{k\ell}$ and $B_{k\ell}$, which get restricted to $x_{k\ell}\vee z$ and $\overline{x_{k\ell}}\vee \overline{z}$ respectively.
	
	Let $\alpha_0 = \sigma \cup \{x_{k\ell}=0\}$ and $\alpha_1 = \sigma \cup \{x_{k\ell}=1\}$. Since both $\alpha_0$ and $\alpha_1$ extend $\rho_p$, they follow path $p$, therefore $h(\alpha_0) = h(\alpha_1)$.
	If $h(\alpha_0) = h(\alpha_1) = 0$, then
	$(\alpha_1, h(\alpha_1))$ satisfies all clauses of $\CR_n$. On the other hand, if $h(\alpha_0) = h(\alpha_1) = 1$, then 
	$(\alpha_0, h(\alpha_0))$ satisfies all clauses of $\CR_n$. Thus in either case, $h$ is not a countermodel for $\CR_n$. 
\end{proof}

From \Cref{thm:lb-from-strategy-extraction} and \Cref{lem:treeMRes-CR}, we obtain the desired lower bound. 
\begin{theorem}
	\label{thm:treeMRes-CR}
	Every tree-like $\MRes$ refutation of $\CR_n$ formulas has size at least $2^n$.
\end{theorem}

\subsubsection{Lower bound for regular $\MRes$}

We now show that these formulas require exponential-size refutations
in regular $\MRes$.  The idea of
using \Cref{thm:lb-from-strategy-extraction} does not work here, since
there is a winning strategy for the universal player with a small
read-once branching program. (The strategy is to choose $z=0$ exactly
if $X$ has an all-0s column.)
We therefore directly analyse the combinatorial structure of a
refutation to show that it must be large.

\begin{theorem}
	\label{thm:CR-regular-lower-bound}
	Every $(A\cup B)$-regular refutation of $\CR_n$ in $\MRes$, and hence every regular $\MRes$ refutation,  has size at least 	$2^{n-1}$.
\end{theorem}

The high level idea is the same as the $\LQParity$ lower bound in \Cref{subsec:LQParity-lb}.
Let $\Pi$ be a $(A\cup B)$-regular
\MRes\ refutation of $\CR_n$. Since every axiom has a variable from
$A\cup B$ while the final clause in $\Pi$ is empty, there is a maximal
``component'' (say $\mathcal{S}$) of the proof leading to and including the final line,
where all clauses are $(A\cup B)$-free.  The clauses in this component
involve only the $X$ variables. We show that the ``boundary'' $\partial{\mathcal{S}}$ of this
component is large, by showing in \Cref{lem:CR-regular-width} that
each clause here must be wide.  (This idea was used in
\cite{PeitlSS19a} to show that $\CR$ is hard for reductionless
$\LDQRes$.)

To establish the width bound, we first note that except for the axioms
$L_A,L_B$, no lines have trivial strategies.  Since the pivots at the
boundary are variables from $A\cup B$, which are all to the right of
$z$, the merge maps incoming into each boundary resolution must be
isomorphic.  By analysing what axiom clauses cannot be used to derive
lines just above the boundary, we show that many variables are absent
in the corresponding merge maps, and invoking soundness of $\MRes$, we
show that they must then be present in the boundary clause, making it
wide.

\begin{proof}[Proof of \Cref{thm:CR-regular-lower-bound}]
	Let $\Pi$ be an $(A\cup B)$-regular refutation of $\CR_n$ (for $n \geq
	2$) in $\MRes$.
	
	The lines of $\Pi$ will be denoted by $L$, $L_1$, $L_2$, $L'$, $L''$ etc. For lines $L$, $L_i$ and $L'$, and universal variable $z$; the respective clause, merge map and the function computed by the merge map will be denoted by $C$, $M^z$, $h^z$, $C_i$, $M^z_i$, $h^z_i$ and $C'$, $(M')^z$, $(h')^z$ respectively. However, since $\CR_n$ has a single universal variable, we avoid the superscript $z$ for the merge-maps and the corresponding functions. So, for these formulas, they will be denoted simply by $M$ $M_i$, $M'$, $h$, $h_i$ and $h'$.
	
	Define $\mathcal{S}$ to be the set of those lines in $\Pi$ where the
	clause part has no variable from $A\cup B$, and furthermore there is a
	path in $G_\Pi$ from the line to the final empty clause via lines
	where all the clauses also have no variables from $A\cup B$. Let
	$\partial{\mathcal{S}}$, called the boundary of $\mathcal{S}$, denote the set of leaves in the subgraph of $G_\Pi$
	restricted to $\mathcal{S}$; these are lines that are in
	$\mathcal{S}$ but their parents are not in $\mathcal{S}$.  Note that
	no leaf of $\Pi$ is in $\mathcal{S}$ because all leaves of $G_\Pi$
	contain a variable in $A\cup B$.

	By definition, for each $L=(C,M) \in \mathcal{S}$, we have  $\var(C) \subseteq X$.
	The sub-derivation $\widehat{\Pi} = \{ C \mid L=(C,M) \in \mathcal{S}
	\}$ contains a propositional resolution refutation of the conjunction
	of clauses $\mathcal{C} = \{ C \mid L=(C,M) \in \partial{\mathcal{S}}\}$. Hence $\mathcal{C}$ is
	an unsatisfiable CNF formula over the $n^2$ variables in $X$. We show
	below, in \Cref{lem:CR-regular-width}, that each clause in $\mathcal{C}$ has
	width at least $n-1$. Hence it is falsified by at most $2^{n^2-(n-1)}$
	assignments. Therefore, to ensure that each of the $2^{n^2}$
	assignments falsifies some clause, at least $2^{n-1}$ clauses are
	required. Therefore $\card{\mathcal{C}} \geqslant 2^{n-1}$. Hence $\card{\Pi} =
	2^{\Omega(n)}$.
\end{proof}
	
	To conclude the proof of \Cref{thm:CR-regular-lower-bound}, it remains to prove the following lemma:
\begin{lemma}
	\label{lem:CR-regular-width}
	For all $L = (C, M) \in \partial{\mathcal{S}}$, $\width(C) \geq n-1$.
\end{lemma}
\begin{proof}
	Since $\var(C) \cap (A \cup B) = \emptyset$, we know that $L$ is not a leaf of $\Pi$. Say $L = \res(L_1,L_2,v)$ where $L_1 = (C_1, M_1)$ and $L_2 = (C_2, M_2)$. Since $\var(C_1) \cap (A \cup B) \neq \emptyset$ and $\var(C_2) \cap (A \cup B) \neq \emptyset$, we have $v \in A \cup B$. Consider the case when $v\in A$; the argument for the case when $v\in B$ is symmetric. Without loss of generality, assume that $v = a_n$, and $a_n \in C_1$ and $\overline{a_n} \in C_2$.

	Since $\Pi$ is $(A\cup B)$-regular, $a_n$ does not occur as a pivot
	in the sub-derivation $\Pi_{L_1}$. Therefore $L_A \not\in
	\leaves(G_{\Pi_{L_1}})$ (otherwise $\overline{a_n} \in C_1$, and
	therefore $C_1$ would be tautological clause, a contradiction).
	This implies that the sub-derivation $\Pi_{L_1}$ cannot use any
	axiom that contains a positive literal in $A$ other than $a_n$, since
	such a literal would have to be eliminated by resolution before
	reaching $C_1$, requiring the corresponding negated literal, and
	$L_A$ is the only axiom with negated literals from $A$. That is,
	$\Pi_{L_1}$ does not use any of the axioms $A_{ij}$ for $i\in[n-1]$. The
	positive literal $x_{ij}$ appears only in $A_{ij}$. Hence for $i\in[n-1]$,
	$j\in[n]$, $x_{ij}$ is not a pivot in $\Pi_{L_1}$ and hence does not
	appear in $M_1$. On the other hand, $M_1$ is not trivial since
	some $A_{nj}$ clause is used.

	The clause $C_2$ contains $\overline{a_n}$, but no other $\overline{a_i}$. So
	$C_2$ is not the axiom $L_A$. Hence $M_2$ is not trivial.
	
	Since the pivot $a_n$ at the step obtaining line $L$ is to the right
	of $z$, by the rules of $\MRes$, $M_1$ and $M_2$ are
	isomorphic. Hence for each $i\in[n-1]$, and each $j\in[n]$, $x_{ij}
	\not\in \var(M_2)$.  We claim the following:
	\begin{claim}
		Either for all $i\in[n-1]$, $C_2$ has a variable of the form
		$x_{i*}$, or for all $j\in[n]$, $C_2$ has a variable of the
		form $x_{*j}$. In either case, $C_2$ has at least $n-1$ variables.
	\end{claim}
		We know that $\overline{a_n}\in C_2$, and for all $i\in[n-1]$, for all $j\in
		[n]$, $x_{ij}\not\in\var(M_2)$. Aiming for contradiction, suppose that
		there exist $i\in[n-1]$ and $j\in[n]$ such that for all $\ell \in [n]$,
		$x_{i\ell}\not\in \var(C_2)$, and for all $k\in[n]$,
		$x_{kj} \not\in \var(C_2)$.
		Fix such an $i,j$.
		
		Let $\rho$ be the minimum partial assignment falsifying $C_2$.
		Then \begin{itemize}
			\item $\rho$ sets $a_n=1$, leaves all other variables in $A\cup B$ unset.
			\item $\rho$ does not set any $x_{i\ell}$ or $x_{kj}$.
		\end{itemize}
		For $c\in \{0,1\}$, extend $\rho$ to $\alpha_c$ as follows: Set
		$a_i=0, b_j=0$, set all other unset variables from $A\cup B$ to 1.
		Set $x_{ij}=c$.  All $x_{i\ell}$ other than $x_{ij}$ set to 1.
		All $x_{kj}$ other than $x_{ij}$ set to 0.  Set remaining
		variables arbitrarily (but in the same way in $\alpha_0$ and
		$\alpha_1$).
		
		The common part of $\alpha_0$ and $\alpha_1$ satisfies
		all axiom clauses except $A_{ij}$ and $B_{ij}$, and does not
		falsify any axiom. The extensions $\alpha_c$ satisfy one more
		axiom, and still do not falsify the remaining axiom (it has a
		universal literal $z$ or $\overline{z}$). They both falsify $C_2$,
		since they extend $\rho$.
		
		Since $\alpha_0$ and $\alpha_1$ agree everywhere except on
		$x_{ij}$, and since $x_{ij}\not\in\var(M_2)$, it follows that
		$M_2(\alpha_0) = M_2(\alpha_1) = d$, say. 
		
		By \Cref{ref:soundness}, both $(\alpha_0,d)$ and $(\alpha_1,d)$
		should falsify some axiom.  However, $(\alpha_{\bar{d}},d)$
		actually satisfies all axioms, a contradiction.
This completes the proof of the claim, and hence of the lemma as well.
\end{proof}

\subsection{The \texorpdfstring{$\KBKFlq$}{KBKF-lq} formulas}
\label{subsec:KBKFlq-lb}
In this \namecref{subsec:KBKFlq-lb}, we turn towards 
the $\KBKFlq$ formulas, defined in \cite{BWJ-SAT14}. 
These formulas are variants of the $\KBKF$ formulas defined in \cite{KBKF95}.
The $\KBKF$ formulas and their variants are significant in QBF proof complexity. The $\KBKF$ formulas  were  used to prove the first lower bound in QBF proof complexity, for the proof system $\QRes$ \cite{KBKF95}.
 They have short refutations in  $\LDQRes$ \cite{ELW13} (that paper uses the name $\varphi_t$) and in $\QURes$ \cite{Gelder12}. 
Variants of this formula were then used to show non-simulations beteen $\QURes$, $\LDQRes$, $\IR$, and others. The specific variant of interest to us here is 
 $\KBKFlq$. This variant was constructed in \cite{BWJ-SAT14} to obtain formulas hard for $\LDQRes$.  It is known that the $\KBKFlq$ formulas are hard for $\LDQRes$ \cite{BWJ-SAT14} and for $\IRM$ \cite{BCJ-ToCT-Sep2019} but have polynomial-size refutations in $\QURes$ \cite{BWJ-SAT14}.

Here we show that the $\KBKFlq$ formulas are hard for the full system of Merge Resolution, thus making it our strongest lower bound in the paper. This constitutes the first genuine-to-QBF lower bound for unrestricted $\MRes$ in the literature.

\begin{definition}
	The $\KBKFlq_n$ formulas \cite{BWJ-SAT14} consist of the quantifier prefix
	\[\exists d_1, e_1, \forall x_1, \exists d_2, e_2, \forall x_2, \ldots, \exists d_n, e_n, \forall x_n, \allowbreak \exists f_1, f_2, \ldots, f_n\]
	and the clauses
	\begin{IEEEeqnarray*}{lClClClCr}
		A_0 &=& \{\overline{d_1}, \overline{e_1}, \overline{f_1}, \ldots, \overline{f_n}\}\\
		A^d_i &=& \{d_i, x_i, \overline{d_{i+1}}, \overline{e_{i+1}}, \overline{f_1}, \ldots, \overline{f_n}\}& \quad &
		A^e_i &=& \{e_i, \overline{x_i}, \overline{d_{i+1}}, \overline{e_{i+1}}, \overline{f_1}, \ldots, \overline{f_n}\} & \qquad & \forall i \in [n-1]\\
		A^d_n &=& \{d_n, x_n, \overline{f_1}, \ldots, \overline{f_n}\} & \quad &
		A^e_n &=& \{e_n, \overline{x_n}, \overline{f_1}, \ldots, \overline{f_n}\}\\
		B^0_i &=& \{x_i, f_i, \overline{f_{i+1}}, \ldots \overline{f_n} \} &\quad &
		B^1_i &=& \{\overline{x_i}, f_i, \overline{f_{i+1}}, \ldots \overline{f_n} \} & \qquad & \forall i \in [n-1]\\
		B^0_n &=& \{x_n, f_n\} & \quad &
		B^1_n &=& \{\overline{x_n}, f_n\}
	\end{IEEEeqnarray*}
\end{definition}

Note that the existential part of each clause in $\KBKFlq_n$ is a
Horn clause (at most one positive literal), and except $A_0$, is even
strict Horn (exactly one positive literal).

We use the following shorthand notation.
Sets of variables: $D =  \{d_1, \ldots, d_n\}$, $E =  \{e_1, \ldots, e_n\}$, $F =  \{f_1, \ldots, f_n\}$, and 
$X =  \{x_1, \ldots, x_n\}$.
Sets of literals: For $Y\in \{D,E,X,F\}$, set 
$Y^1 = \{ u \mid u\in Y\}$ and
$Y^0 = \{ \overline{u} \mid u\in Y\}$. 
Sets of clauses:
\[\begin{array}[t]{lclclcl}
	\mathcal{A}_0 &=& \{A_0\}\\
	\mathcal{A}_i &=& \{A^d_i, A^e_i\} \quad \forall i \in [n] & \quad \quad & \mathcal{B}_i &=& \{B^0_i, B^1_i\} \quad \forall i \in [n]\\
	\mathcal{A}_{[i,j]} &=& \cup_{k \in [i,j]} \mathcal{A}_k \quad \forall i,j \in [0,n], i \le j & \quad & \mathcal{B}_{[i,j]} &=& \cup_{k \in [i,j]} \mathcal{B}_k \quad \forall i,j \in [n], i \le j\\
	\mathcal{A} &=& \mathcal{A}_{[0,n]} & \quad & \mathcal{B} &=& \mathcal{B}_{[1,n]} 
\end{array}
\]

\begin{theorem}
  \label{thm:mres-lb}
Every $\MRes$ refutation of $\KBKFlq_n$ has size at least $2^n$.
\end{theorem}

This proof follows the same high-level idea as the proofs of \Cref{thm:LQParity-lb,thm:CR-regular-lower-bound}. Namely, in any refutation, a maximal component is identified (in this case, the $F$-free component, where no clause has a variable from $F$) and its boundary is shown to be large. However, the idea for showing that the boundary is large is completely different. The proofs of 
\Cref{thm:LQParity-lb,thm:CR-regular-lower-bound} established that the boundary  clauses must be wide. Here, the complexity measure for the boundary is not width but the nature of the merge maps, or equivalently, of the partial strategies. We identify a property called {\em self-dependence} which captures the right complexity; a merge map for $x_i$ has this self-dependence property if it depends on at least one of $d_i,e_i$. We show that all merge-maps at the final line must have self-dependence, whereas at the  boundary lines none of the merge maps have self-dependence. 
We use this to then  conclude that there must be exponentially many lines.

To show that self-dependence is not possible outside the $F$-free
component, we show that from a line with $F$-variables and at least
one self-dependent strategy, the $F$-variables can never be removed.

Elaborating on the roadmap of the argument: 
Let $\Pi$ be an $\MRes$ refutation of $\KBKFlq_n$.  Each line in $\Pi$
has the form $L=(C,M^{x_1},\ldots,M^{x_n})$ where $C$ is a clause over
$D,E,F$, and each $M^{x_i}$ is a merge map computing a strategy for $x_i$.

Define $\mathcal{S}$ to be the set of those lines in $\Pi$ where the
clause part has no $F$ variable and furthermore the line has a path in
$G_\Pi$ to the final empty clause via lines where all the clauses also
have no $F$ variables. Let $\partial{\mathcal{S}}$, called the boundary of $\mathcal{S}$, denote the set of leaves in
the subgraph of $G_\Pi$ restricted to $\mathcal{S}$; these are lines
that are in $\mathcal{S}$ but their parents are not in
$\mathcal{S}$.  Note that by definition, for each $L=(C,\{M^{x_i}\mid
i\in[n]\}) \in \mathcal{S}$, $\var(C) \subseteq D \cup E$. No line in
$\mathcal{S}$ (and in particular, no line in $\partial{\mathcal{S}}$) is an
axiom since all axiom clauses have variables from $F$.

Recall that the variables of $\KBKFlq_n$ can be naturally grouped
based on the quantifier prefix: for $i\in[n]$, the $i$th group has
$d_i,e_i, x_i$, and the $(n+1)$th group has the $F$ variables. By
construction, the merge map for $x_i$ does not depend on variables in
later groups, as is indeed required for a countermodel.  We say that a
merge map for $x_i$ has {\em self-dependence} if it does depend on $d_i$
and/or $e_i$.

We show that every merge map at every line in $\mathcal{S}$ is non-trivial
(\Cref{lem:S'-strategies-nontrivial}).  Further, we show that at
every line on the boundary of $\mathcal{S}$, i.e.~in $\partial{\mathcal{S}}$, no merge map has
self-dependence (\Cref{lem:boundary-strategies}). Using this, we
conclude that $\partial{\mathcal{S}}$ must be exponentially large, since in every
countermodel the strategy of each variable must have self-dependence
(\Cref{prop:KBKF-strategies}).

In order to show that lines in $\partial{\mathcal{S}}$ do not have self-dependence, we first
establish several properties of the sets of axiom clauses used in a
sub-derivation
(\Cref{lem:posFliterals-and-uci,lem:empty-uci-structure,lem:nonempty-uci-structure,lem:self-dependent-strategy-suffix-interval}).

For a line $L \in \Pi$, let $\Pi_{L}$ be the minimal sub-derivation of
$L$, and let $G_{\Pi_{L}}$ be the corresponding subgraph of $G_{\Pi}$
with sink $L$. Let $\UsedConstraintInd(\Pi_{L}) = \{ i \in [0,n] \mid
\leaves(G_{\Pi_{L}}) \cap \mathcal{A}_i \neq
\emptyset\}$. ($\UsedConstraintInd$ stands for
UsedConstraintsIndex). Note that we are only looking at the clauses in
$\mathcal{A}$ to define $\UsedConstraintInd$.

\begin{lemma}
\label{lem:posFliterals-and-uci}
For every line $L = (C, \{M^{x_i} \mid i \in [n]\})$ of $\Pi$, 
$\card{C\cap F^1}\le 1$. Furthermore,
$\UsedConstraintInd(\Pi_L) = \emptyset\Leftrightarrow C \cap F^1 \neq \emptyset$. 
(Recall that $F^1 = \{f_1, f_2, \ldots, f_n\}$, i.e.~the set of positive literals over the variable set $\{f_1, f_2, \ldots, f_n\}$.)
\end{lemma}
\begin{proof}
	Since the existential part of each clause in $\KBKFlq_n$ is a Horn
	clause, and since the resolvent of Horn clauses is also Horn,
	$\card{C\cap F^1}\le 1$ for each line of $\Pi$. It thus suffices to
	prove that $\forall L \in \Pi, ~~\UsedConstraintInd(\Pi_L) = \emptyset
	\iff C\cap F^1 \neq \emptyset$.
	
	($\Rightarrow$): For an arbitrary line $L\in \Pi$, suppose
	$\UsedConstraintInd(\Pi_L) = \emptyset$, so $L$ is derived from
	$\mathcal{B}$. Since $\var_\exists(\mathcal{B})=F$, $\var(C) \subseteq
	F$. The existential part of these clauses is strict Horn, and the
	resolvent of strict Horn clauses is also strict Horn, so $C$ is strict
	Horn. So $C\cap F^1 \neq \emptyset$.
	
	($\Leftarrow$): The statement $C\cap F^1 \neq \emptyset \Rightarrow
	\UsedConstraintInd(\Pi_L) = \emptyset$ holds at all axioms.  Assume to
	the contrary that it does not hold everywhere in $\Pi$.  Pick a
	highest $L$ (closest to the axioms) for which this statement
	fails. That is, $C\cap F^1 \neq \emptyset$, and
	$\UsedConstraintInd(\Pi_L) \neq \emptyset$. Let $L',L''$ be the
	parents of $L$ in $\Pi$; by choice of $L$, both $L'$ and $L''$ satisfy
	the statement. Let $f_j$ be the positive literal in $C$ (unique,
	because $C$ is Horn).  Without loss of generality, $f_j\in C'$.  Since
	$L'$ satisfies the statement, $\UsedConstraintInd(\Pi_{L'}) =
	\emptyset$. So $\var(C')\subseteq F$, and since $C'$ is Horn, $C'
	\setminus \{f_j\} \subseteq F^0$. Since $f_j\in C$, the pivot at this
	step is not $f_j$, so it must be an $f_k$ for some $\overline{f_k}\in
	C'$. So $f_k\in C''$.  Since $L''$ satisfies the statement,
	$\UsedConstraintInd(\Pi_{L''}) = \emptyset$. But then 
	$\UsedConstraintInd(\Pi_{L}) = \UsedConstraintInd(\Pi_{L'})\cup
	\UsedConstraintInd(\Pi_{L''}) = \emptyset$, contradicting our choice
	of $L$. Hence our assumption was wrong, and the statement holds for
	all $L$ in $\Pi$.
\end{proof}

\begin{lemma}
  \label{lem:empty-uci-structure}
A line $L = (C, \{M^{x_i} \mid i \in [n]\})$ of $\Pi$ with
$\UsedConstraintInd(\Pi_L) = \emptyset$ has these properties:
\begin{enumerate}
\item \label{item:empty-uci-vars}
  $\var(C) \subseteq F$; for all  $i\in[n]$, $M^{x_i} \in \{*,0,1\}$;
\item \label{item:empty-uci-positive-map}
  For some $j\in [n]$, $f_j \in C$ and $M^{x_j} \in \{0,1\}$; such a $j$ is unique; 
\item \label{item:empty-uci-maps-left-of-positive}
  For the unique $j$ from (\ref{item:empty-uci-positive-map}), for $1 \le i < j$, $f_i \not\in \var(C)$ and $M^{x_i} = *$;
\item \label{item:empty-uci-maps-right-of-positive}
  For $j < i \le n$,
  if $f_i\not\in \var(C)$, then $M^{x_j} \in \{0,1\}$.
\end{enumerate}
\end{lemma}
\begin{proof}
	\begin{enumerate}
		\item Since $\UsedConstraintInd(\Pi_L) = \emptyset$,
		$\var(C) \subseteq \var_\exists(\mathcal{B})=F$.
		
		All pivots in $\Pi_L$ are from $F$, and all universal variables are
		left of $F$ in the quantifier prefix. So no step in $\Pi_L$ can use
		the merge operation to update merge maps; all steps in $\Pi_L$ use
		only the select operation, which does not create any branching.
		
		\item By \Cref{lem:posFliterals-and-uci}, $\card{C\cap F^1}=1$,
		so there is a unique $j$ with the literal $f_j \in C$. This literal
		appears only in the clauses of $\mathcal{B}_j$, both of which create
		a non-trivial strategy for $x_j$. So $M^{x_j} \neq *$. By
		item~(\ref{item:empty-uci-vars}) proven above, $M^{x_j} \in
		\{0,1\}$.
		
		\item 
		Let $k$ be the least index such that $\Pi_L$ uses an axiom from
		$\mathcal{B}_k$. Since the positive literal $f_j$ is in $C$ and
		appears only in $\mathcal{B}_j$, $k \le j$. Assume $k < j$.  The
		axiom from $\mathcal{B}_k$ introduces the positive literal $f_k$
		into $\Pi_L$, and by choice of $k$, no axiom in $\Pi_L$ has the
		literal $\overline{f_k}$. Hence $f_k$ cannot be removed by
		resolution, and so $f_k\in C$, contradicting the fact that $C$ is
		Horn. So in fact $k=j$. This means that no axiom introduces the
		variables $f_i$, $i< j$, into $\Pi_L$, so $f_i\not\in\var(C)$.
		Furthermore, amongst all the axioms in $\mathcal{B}$, only the
		axioms in $\mathcal{B}_i$ have a non-trivial merge map for
		$x_i$. Hence for $i< j$, no non-trivial merge map for $x_i$ is
		created.
		
		\item   Since $f_j\in C$, $\Pi_L$ uses an axiom from $\mathcal{B}_j$. This
		axiom introduces the literals $\overline{f_i}$, for $j < i \le n$,
		into $\Pi_L$.
		
		If $\overline{f_i}$ is removed (by resolution) in $\Pi_L$, then an
		axiom from $\mathcal{B}_i$ must be used to introduce the positive
		literal $f_i$. This axiom created a non-trivial merge map for $x_i$,
		so the merge map for $x_i$ at $L$ is also non-trivial. \qedhere
	\end{enumerate}
	
\end{proof}

\begin{lemma}  
  \label{lem:nonempty-uci-structure}  
  Let $L = (C, \{M^{x_i} \mid i \in [n]\})$ be a line of $\Pi$ with
  $\UsedConstraintInd(\Pi_L) \neq \emptyset$.  Then
  $\UsedConstraintInd(\Pi_{L})$ is an interval $[a,b]$ for some $0 \le
  a \le b \le n$. Furthermore, {\upshape(}in the items below, $a,b$ refer to the
  endpoints of this interval\,{\upshape)}, it 
  has the following properties:  
\begin{enumerate}
\item \label{item:no-*-in-uci}
  For $k\in[n]\cap [a,b]$, $M^{x_k}\neq *$.
\item \label{item:start-uci-positive-literals}
  If $a \ge 1$, then $\card{\{d_a,e_a\}\cap C}=1$.
  If $a=0$, then $C$ does not have any positive literal.
\item \label{item:end-uci-negated-literals}
  If $b < n$, then $\overline{d_{b+1}}, \overline{e_{b+1}} \in C$.
\item \label{item:outside-uci-mergemaps}
  For all $k \in [n]\setminus[a,b]$, {\upshape(i)}~$d_k, e_k \not\in \var(M^{x_{k}})$, and 
  {\upshape(ii)}~if $M^{x_k} = *$ then $\overline{f_k} \in C$.
\end{enumerate}
\end{lemma}
\begin{proof}
	Assume to the contrary that $\UsedConstraintInd(\Pi_L)$ is not an
	interval. Then there exist $0\le a < c < b \le n$ such that $a,b \in
	\UsedConstraintInd(\Pi_L)$ but $c \not\in
	\UsedConstraintInd(\Pi_L)$. Let $L_1$ be the first line in $\Pi_{L}$
	such that $\UsedConstraintInd(\Pi_{L_1})$ intersects both $[0,c-1]$
	and $[c+1,n]$ (note that $L_1$ exists). Since leaves have singleton
	$\UsedConstraintInd$ sets, $L_1$ is not a leaf. Say $L_1 =
	\res(L_2,L_3,v)$. By our choice of $L_1$, exactly one each of
	$\UsedConstraintInd(\Pi_{L_2})$ and $\UsedConstraintInd(\Pi_{L_3})$ is
	a non-empty subset of $[0,c-1]$ and of $[c+1,n]$.
	So $v \in \var_{\exists}(\mathcal{A}_{[0,c-1]})$ and $v \in
	\var_{\exists} (\mathcal{A}_{[c+1,n]})$.  But
	$\var_{\exists}(\mathcal{A}_{[0,c-1]}) \cap
	\var_{\exists}(\mathcal{A}_{[c+1,n]}) = F$, and by
	\Cref{lem:posFliterals-and-uci}, both $C_2$ and $C_3$ contain
	variables of $F$ only in negated form. So no variable from $F$ can be
	a resolution pivot, a contradiction. It follows that
	$\UsedConstraintInd(\Pi_L)$ is an interval.
	
	\begin{enumerate}
		\item For $k\in[n]\cap[a,b]$, some axiom from $\mathcal{A}_k$ has been
		used to derive $L$. Both these axioms create non-trivial strategies
		for $x_k$. Subsequent $\MRes$ steps cannot make a non-trivial
		strategy trivial.
		
		\item Consider first the case $a \ge 1$.  Since $C$ is a Horn clause,
		$C$ can contain at most one of the literals $d_a,e_a$.
		
		Since $a \in \UsedConstraintInd(\Pi_L)$, at least one of $A^d_a,
		A^e_a$ appears in $\leaves(\Pi_{L})$, so at least one of the
		literals $d_a,e_a$ is introduced into $\Pi_L$.  Since $A^d_{a-1}$
		and $A^e_{a-1}$ are the only axioms that contain $\overline{d_a}$ or
		$\overline{e_a}$, and since neither of these is used in $\Pi_L$,
		therefore the positive literals $d_a,e_a$, if introduced, cannot be
		removed through resolution. Hence at least one of them is in $C$.
		It follows that $C$ has exactly one of $d_a,e_a$.
		
		If $a=0$, $\Pi_L$ uses the clause $A_0$ which has only negative
		literals. The resolvent of such a clause and a Horn clause also has
		only negative literals. Following the sequence of resolutions on the
		path from a leaf using $A_0$ to $C$ shows that $C$ has only negative
		literals.
		
		\item Since $b<n$ and $b \in \UsedConstraintInd(\Pi_L)$, some clause
		from $\mathcal{A}_b$ is used in $\Pi_L$ and introduces the literals
		$\overline{d_{b+1}}, \overline{e_{b+1}}$ into $\Pi_L$. Since $b+1
		\not\in \UsedConstraintInd(\Pi_L)$, no leaf of $\Pi_{L}$ contains
		the positive literals $d_{b+1},e_{b+1}$. So $\overline{d_{b+1}}$
		and $\overline{e_{b+1}}$ cannot be removed through resolution.
		
		\item For $k>b$, no leaf in $\Pi_L$ contains the positive literals
		$d_k,e_k$. For $k < a$, no leaf in $\Pi_L$ contains the negative
		literals $\overline{d_k},\overline{e_k}$.  Thus, for
		$k\not\in[a,b]$, the variables $d_k,e_k$ are not used as resolution
		pivots anywhere in $\Pi_L$, and hence are not queried in any of the
		merge maps.
		
		Each negative literal $\overline{f_k}$ is present in every clause of
		$\mathcal{A}$, and hence is introduced into $\Pi_L$.  If $M^{x_k} =
		*$, then $B^0_k, B^1_k \not\in \leaves(\Pi_{L})$ (both of them have
		non-trivial merge maps for $x_k$). Since these are the only clauses
		with the positive literal $f_k$, the literal $\overline{f_k}$ cannot
		be removed in $\Pi_L$; hence $\overline{f_k} \in C$.
		\qedhere
	\end{enumerate}
\end{proof}

\begin{lemma}
  \label{lem:self-dependent-strategy-suffix-interval}
  For any line $L=(C,\{M^{x_i}\mid i\in[n]\})$ in $\Pi$, and any
  $k\in[n]$, if $\{d_k,e_k\} \cap \var(M^{x_k}) \neq \emptyset$, then 
  $\UsedConstraintInd(\Pi_{L}) =[a,n]$ for some $a\le k-1$.
\end{lemma}
\begin{proof}
	Since $\{d_k,e_k\} \cap \var(M^{x_k}) \neq \emptyset$, either $d_k$
	or $e_k$ must be used as a pivot in $\Pi_L$, and hence must appear
	in both polarities in $\Pi_L$. The variables $d_k,e_k$ appear
	positively only in $\mathcal{A}_k$, and negatively only in
	$\mathcal{A}_{k-1}$. Hence $a \le k-1$. 
	
	Suppose $b < n$. By \Cref{lem:nonempty-uci-structure}
	(\ref{item:end-uci-negated-literals}), both $\overline{d_{b+1}}$ and
	$\overline{e_{b+1}}$ are in $C$. Consider any path $\rho$ in $\Pi$
	from $L$ to the final line $L_\Box$. At every line on this path, the
	merge map for $x_k$ queries at least one of $d_k,e_k$ since it is
	at least as complex as the merge map $M^{x_k}$. Along this path, both
	$d_{b+1}$ and $e_{b+1}$ must appear as pivots, since the negated
	literals are eventually removed. Pick the first such step on $\rho$,
	and assume without loss of generality that the pivot is $d_{b+1}$
	(the other case is symmetric). So $\overline{d_{b+1}}$ is present in
	the line, say $L_1$, on $\rho$, and $d_{b+1}$ is present in the
	clause $L_2$ with which it is resolved to obtain $L_3 =
	\res(L_2,L_1,d_{b+1})$ on $\rho$. By
	\Cref{lem:nonempty-uci-structure}~(\ref{item:start-uci-positive-literals}),
	$\UsedConstraintInd(\Pi_{L_2}) =[b+1,b']$ for some $b' \ge b+1$.
	Hence by
	\Cref{lem:nonempty-uci-structure}~(\ref{item:outside-uci-mergemaps}),
	$d_k,e_k \not\in \var(M_2^{x_k})$.  However, $\{d_k,e_k\} \cap
	\var(M_1^{x_k}) \neq \emptyset$. Since this resolution on
	$d_{b+1}$ is not blocked, it must be the case that $M_2^{x_k} =
	*$. Hence, by
	\Cref{lem:nonempty-uci-structure}~(\ref{item:outside-uci-mergemaps}),
	$\overline{f_k} \in C_2$ and so $\overline{f_k} \in C_3$.  To remove
	this literal, at some later point along $\rho$, $f_k$ must appear as
	pivot. However, at that point, the line from $\rho$ has a complex
	merge map for $x_k$, while the line with the positive literal $f_k$
	has a non-trivial constant merge map (by
	\Cref{lem:empty-uci-structure}~(\ref{item:empty-uci-positive-map})). Hence the resolution on $f_k$
	is blocked, a contradiction. It follows that $b=n$. 
\end{proof}

\begin{lemma}
\label{lem:S'-strategies-nontrivial}
  For all $L \in \mathcal{S}$, for all $k \in [n]$, $M^{x_{k}}
  \neq *$. 
\end{lemma}
\begin{proof}
  Consider a line $L=(C,\{M^{x_i}\mid i\in[n]\})\in\mathcal{S}$.
  Since $L \in \mathcal{S}$, it has no variables from $F$. So
  $C\cap F^1 = \emptyset$. (Recall that $F^1$ is the set of positive literals with variables from $F$; that is,
  $F^1 = \{f_1, f_2, \ldots, f_n\}$. Similarly, $F^0 = \{\overline{f_1}, \overline{f_2}, \ldots, \overline{f_n}\}$ is the set of negative literals over $F$.) By
  \Cref{lem:posFliterals-and-uci},
  $\UsedConstraintInd(\Pi_L) \neq \emptyset$. Since every clause in
  $\mathcal{A}$ contains all literals in $F^0$, for each $k\in[n]$,
  $\Pi_L$ has a leaf where the clause contains $\overline{f_k}$. This
  literal is removed in deriving $L$, so $\Pi_L$ also has a leaf where
  the clause contains the positive literal $f_k$. That is, it uses an
  axiom from $\mathcal{B}_k$; this leaf has a non-trivial merge map
  for $x_k$. Since a step in $\MRes$ cannot make a non-trivial merge
  map trivial, the merge map for $x_k$ at $L$ is non-trivial.
\end{proof}

\begin{lemma}
\label{lem:boundary-strategies}
For all $L \in \partial{\mathcal{S}}$, for all $k \in [n]$,
$d_k, e_k \not\in \var(M^{x_{k}})$. 
\end{lemma}

\begin{proof}  
  Consider a line $L \in \partial{\mathcal{S}}$; $L=(C,\{M^{x_i}\mid i\in[n]\})$.
  Assume to the contrary that for some $k\in[n]$, $\{d_k, e_k\} \cap
  \var(M^{x_{k}}) \neq \emptyset$.

  The line $L$ is obtained by performing resolution on two
  non-$\mathcal{S}$ clauses with a pivot from $F$. Let
  $L=\res(L',L'',f_\ell)$ for some $\ell\in [n]$; $f_\ell\in C'$ and
  $\overline{f_\ell}\in C''$. Since $L$ has no variable in $F$,
  $f_\ell$ is the only variable from $F$ in $\var(C')$ and
  $\var(C'')$.

  Since $C'$ has the literal $f_\ell \in F^1$, by
  \Cref{lem:posFliterals-and-uci},
  $\UsedConstraintInd(\Pi_{L'}) = \emptyset$ and $L'$ is derived
  exclusively from $\mathcal{B}$. Since $D\cup E$ and
  $\var(\mathcal{B})$ are disjoint, all the merge maps in $L'$ have no
  variable from $D\cup E$. So $M^{x_k}$ gets its $D\cup E$ variables
  from $(M'')^{x_k}$. Since this does not block the resolution step,
  $(M')^{x_k}$ must be trivial and $M^{x_k} = (M'')^{x_k}$.  Since
  $\var(C')\cap F = f_\ell$, by
  \Cref{lem:empty-uci-structure}~(\ref{item:empty-uci-positive-map}),(\ref{item:empty-uci-maps-left-of-positive}),(\ref{item:empty-uci-maps-right-of-positive}),
  $k < \ell$.

  The line $L''$ has no literal from $F^1$, so by
  \Cref{lem:posFliterals-and-uci},
  $\UsedConstraintInd(\Pi_{L''}) \neq \emptyset$. It has a merge map
  for $x_k$ involving at least one of $d_k,e_k$, so by
  \Cref{lem:self-dependent-strategy-suffix-interval},
  $\UsedConstraintInd(\Pi_{L''}) =[a,n]$ for some $a\le k-1$.
  Thus we have $a \le k-1 < k < \ell \le n$.

  Consider the resolution of $L'$ with $L''$. By
  \Cref{lem:empty-uci-structure}~(\ref{item:empty-uci-positive-map}),
  $(M')^{x_\ell} \in\{0,1\}$, and by
  \Cref{lem:nonempty-uci-structure}~(\ref{item:no-*-in-uci}),
  $(M'')^{x_\ell} \neq *$. To enable this resolution, $(M'')^{x_\ell}
  = (M')^{x_\ell}$. The clauses $A_\ell^d$ and
  $A_\ell^e$ give rise to different constant strategies for
  $x_\ell$. So the derivation of $L''$ uses exactly one of these two
  clauses. Assume it uses $A_\ell^d$; the other case is
  symmetric. %
  Since $a < \ell$,
  the derivation of $L''$ uses a clause from $A_{\ell-1}$,
  introducing literals $\overline{d_{\ell}}$ and
  $\overline{e_{\ell}}$. Since the only clause containing positive
  literal $e_\ell$ is not used, $\overline{e_{\ell}}$ survives in
  $C''$. Going from $L''$ to $L$ removes only $\overline{f_{\ell}}$,
  so $\overline{e_{\ell}}\in C$.
 
   To summarize, at this stage we know that $L\in \partial{\mathcal{S}}$,
   $\overline{e_\ell}\in C$, $\{d_k,e_k\}\cap \var(M^{x_k}) \neq
     \emptyset$, $M^{x_\ell}\in \{0,1\}$ and $1\le k < \ell \le n$.
  
  Fix any path $\rho$  in $G_\Pi$ from $L$  to $L_\Box$.
  Along this path, $e_\ell$ appears as the pivot somewhere, since the
  literal $\overline{e_\ell}$ is eventually removed. Consider the
  resolution step at that point, say $C_1=\res(C_2,C_3,e_\ell)$, with
  $C_3$ being the clause at the line on $\rho$. At the corresponding
  line $L_3$, the strategies are at least as complex as those at $L$.
  Hence $\var(M_3^{x_k})\cap \{d_k,e_k\} \neq \emptyset$. On the other
  hand, $C_2$ has the positive literal $e_\ell$. By
  \Cref{lem:nonempty-uci-structure}, for the corresponding line $L_2$,
  $\UsedConstraintInd(\Pi_{L_2}) = [\ell,c]$ for some $c \ge
  \ell$. Since $k <\ell $, by \Cref{lem:nonempty-uci-structure},
  $\{d_k,e_k \} \cap \var(M_2^{x_k})=\emptyset$. However, the path from
  $L_2$ to $L_1$ and thence to $L_\Box$ along $\rho$ witnesses that
  $L_2\in\mathcal{S}$, so by
  \Cref{lem:S'-strategies-nontrivial}, $M_2^{x_k}\neq *$.  Thus
  $M_2^{x_k}$ and $M_3^{x_k}$ are non-trivial but not isomorphic, and
  this blocks the resolution on $e_\ell$.

  Thus our assumption that $\{d_k,e_k\}\cap \var(M^{x_k}) \neq
  \emptyset$ must be false. The lemma is proved.
\end{proof}

We will also use the following property of $\KBKFlq$ formulas. It
implies that in every countermodel, the strategy for every variable
has self-dependence. This is used, towards the end of the proof of
\Cref{thm:mres-lb}, to show that merge-maps for countermodels must be complex and
large.
\begin{proposition}
	\label{prop:KBKF-strategies}
	Let $h$ be any countermodel for $\KBKFlq_n$.  Let $\alpha$ be any
	assignment to $D$, and $\beta$ be any assignment to $E$. For
	each $i\in [n]$, if $\alpha_j \neq \beta_j$ for all $1\le j
	\le i$, then $h^{x_i}\big((\alpha,\beta)\restriction_{L_Q(x_i)}\big) = \alpha_i$. In particular, if $\alpha_j\ne\beta_j$ for all $j\in[n]$,
	then the countermodel computes $h(\alpha,\beta)=\alpha$.
\end{proposition}
\begin{proof}%
	Let $h$ be any countermodel for $\KBKFlq_n$. For $i \in [n]$, let $\alpha^i$ be an assignment to $\{d_1, \ldots, d_i\}$, and $\beta^i$ be an assignment to $\{e_1, \ldots, e_i\}$. For $j \le i$, let $\alpha^i_j$ (resp.~$\beta^i_j$) be the assignment to $d_j$ (resp.~$e_j$) set by the assignment $\alpha^i_j$ (resp.~$\beta^i_j$). We will show that
	for each $i\in [n]$, if $\alpha^i_j \neq \beta^i_j$ for all $1\le j
	\le i$, then $h^{x_i}(\alpha^i,\beta^i) = \alpha^i_i$.
	This implies the claimed result.    
	
	Fix some $i \in [n]$. Assume to the contrary that $\alpha^i_j \neq \beta^i_j$ for all $1\le j \le i$ and $h^{x_i}(\alpha^i,\beta^i) \neq \alpha^i_i$. We will give a winning strategy for the existential player. Note that all clauses in $\mathcal{A}[0,i-1]$ are satisfied by the partial assignment $(\alpha^i,\beta^i)$. The existential player sets $d_j = e_j = 1$ for all $j > i$ and sets $f_j = 1$ for all $j \in [n]$. This satisfies all the remaining clauses, irrespective of the strategy of the universal player. Therefore the existential player wins. This contradicts the assumption that $h$ is a countermodel for $\KBKFlq_n$.
\end{proof}

Now we have all the required information; we put it together to obtain
the lower bound.
\begin{proof}[Proof of \Cref{thm:mres-lb}]
  Let $\Pi$ be a refutation of $\KBKFlq_n$ in $\MRes$. 
   Let $\mathcal{S},\partial{\mathcal{S}}$ be as defined in the beginning of
   this section.
   Let the final line of $\Pi$ be $L_\Box=(\Box,\{M_\Box^{x_i} \mid i\in
  [n]\})$, and for $i\in [n]$, let $h_i$ be the functions computed by
    the merge map $M_\Box^{x_i}$. By soundness of $\MRes$, the functions
    $\{h_i\}_{i\in[n]}$ form a countermodel for $\KBKFlq_n$.

  For each $a \in \{0,1\}^n$, consider the assignment $\alpha$ to the
  variables of $D\cup E$ where $d_i=a_i$, $e_i = \overline{a_i}$. Call
  such an assignment an anti-symmetric assignment.  Given such an
  assignment, walk from $L_\Box$ towards the leaves of $\Pi$ as far as
  is possible while maintaining the following invariant at each line
  $L = (C,\{ M^{x_i}\mid i\in[n]\})$ along the way:
  \begin{enumerate}
    \item $\alpha$ falsifies $C$, and
    \item for each $i\in[n]$,  $h_i(\alpha) = M^{x_i}(\alpha)$.
  \end{enumerate}
  Clearly this invariant is initially true at $L_\Box$, which is in
  $\mathcal{S}$. If we are currently at a line $L \in \mathcal{S}$
  where the invariant is true, and if $L\not\in \partial{\mathcal{S}}$, then $L$
  is obtained from lines $L'$, $L''$. The resolution pivot in this step
  is not in $F$, since that would put $L$ in $\partial{\mathcal{S}}$. So both
  $L'$ and $L''$ are in $\mathcal{S}$, and the pivot is in $D \cup
  E$. Let the pivot be in $\{d_\ell,e_\ell\}$ for some
  $\ell\in[n]$. Depending on the pivot value, exactly one of $C',C''$
  is falsified by $\alpha$; say $C'$ is falsified.  By
  \Cref{lem:S'-strategies-nontrivial}, for each $i\in[n]$, both
  $(M')^{x_i}$ and $(M'')^{x_i}$ are non-trivial. By definition of the
  $\MRes$ rule,
  \begin{itemize}
    \item For $i< \ell$, $(M')^{x_i}$ and $(M'')^{x_i}$ are isomorphic
      (otherwise the resolution is blocked),
      and $M^{x_i} = (M')^{x_i} = (M'')^{x_i}$.
    \item For $i\ge \ell$, there are two possibilities: \\
      (1)~$(M')^{x_i}$ and $(M'')^{x_i}$ are isomorphic, and
      $M^{x_i} = (M')^{x_i}$.  \\
      (2)~$M^{x_i}$ is a merge of $(M')^{x_i}$
      and $(M'')^{x_i}$ with the pivot variable queried. By definition
      of the merge operation, since $C'$ is falsified by $\alpha$,
      $M^{x_i}(\alpha) = (M')^{x_i}(\alpha)$.
  \end{itemize}
  Thus in all cases, for each $i$, $h_i(\alpha) = M^{x_i}(\alpha)=
  (M')^{x_i}(\alpha)$.
  Hence  $L'$ satisfies the   invariant.

  We have shown that as long as we have not encountered a line in 
  $\partial{\mathcal{S}}$, we can move further. We continue the walk until a
  line in $\partial{\mathcal{S}}$ is reached. We denote the line so reached by
  $P(\alpha)$.  Thus $P$ defines a map from anti-symmetric assignments
  to $\partial{\mathcal{S}}$.

  We now show that the map $P$ is one-to-one. Suppose, to the
  contrary,
  $P(\alpha) = P(\beta) = (C,\{M^{x_i} \mid i\in[n]\})$ for
  two distinct anti-symmetric assignments obtained from $a,b\in
  \{0,1\}^n$ respectively. Let $j$ be the least index in $[n]$ where
  $a_j \neq b_j$.  By \Cref{lem:boundary-strategies}, $M^{x_j}$
  depends only on $\{d_i,e_i \mid i < j\}$, and $\alpha,\beta$ agree
  on these variables. Thus we get the equalities
  $a_j = h_j(\alpha) = M^{x_j}(\alpha) = M^{x_j}(\beta) = h_j(\beta)= b_j$,
where the first and last equalities follow from
\Cref{prop:KBKF-strategies}, the third equality from 
\Cref{lem:boundary-strategies} and choice of $j$, and the second
and fourth equalities by the invariant satisfied at $P(\alpha)$ and
$P(\beta)$ respectively.
  This contradicts $a_j\neq b_j$.

  We have established that the map $P$ is one-to-one. Hence, 
  $\partial{\mathcal{S}}$ has at least as many lines as anti-symmetric
  assignments, so $\card{\Pi} \ge \card{\partial{\mathcal{S}}} \ge 2^n$.
\end{proof}

\section{Relations among proof systems}
In this section, we collect all the separations among proof systems which are implied by the lower bounds in \Cref{sec:lower bounds}.

Since any propositional formula is also a QBF formula and since $\MRes$ degenerates to Resolution on propositional formulas, it follows from propositional proof complexity that $\MRes$ strictly-simulates tree-like and regular $\MRes$, and that tree-like $\MRes$ does not simulate regular or general $\MRes$. Whether, regular $\MRes$ p-simulates tree-like $\MRes$ is unknown. 
Here we observe that the non-simulation of regular and general $\MRes$ by tree-like $\MRes$ is also witnessed by the $\QParity$ formulas (because the $\QParity$ formulas have polynomial-size refutations in regular and general $\MRes$ but require exponential-size refutations in tree-like $\MRes$).

\label{sec:proof-system-realtions}

The following two theorems show that $\MRes$ and its restrictions are incomparable with some other resolution-based QBF proof systems. As observed in \cite{DBLP:journals/jar/BeyersdorffBM21}, one direction of the non-simulation follows from the Equality formulas: these formulas have polynomial-size refutations in tree-like, regular and general $\MRes$ but require exponential-size refutations in $\QRes$, $\QURes$, $\QCP$, $\ExpRes$, and $\IR$.

\begin{theorem}
	\label{cor:tree-and-regular-MRes-incomparable-with-five-systems}
	Tree-like and regular $\MRes$ are incomparable with the tree-like and general versions
	of $\QRes$, $\QURes$, $\QCP$, $\ExpRes$, and
	$\IR$.
\end{theorem}
\begin{proof}
	We showed in \Cref{thm:treeMRes-CR} that the Completion Principle $\CR_n$ requires exponential-size refutations in tree-like $\MRes$. In \Cref{thm:CR-regular-lower-bound}, we showed that it requires exponential-size refutations in regular $\MRes$.
	It has polynomial-size refutations in tree-like QRes \cite{Janota-QRes-and-CDCL-SAT16} (and hence also in $\QURes$ and $\QCP$) and tree-like $\ExpRes$ \cite{Janota-Expansion-vs-QRes-TCS15} (and hence also in $\IR$). (While \cite{Janota-Expansion-vs-QRes-TCS15} does not explicitly mention tree-like or regular refutations, the refutation provided there for $\CR_n$ is tree-like and regular.)
	Therefore, tree-like and regular $\MRes$ do not simulate the tree-like and general versions of $\QRes$, $\QURes$, $\QCP$, $\ExpRes$, and $\IR$.

The other direction of the non-simulation follows from the Equality formulas, as mentioned in \Cref{ex:Equality}.
\end{proof}

\begin{theorem}
	\label{cor:MRes-incomparable}
	$\MRes$ is incomparable with $\QURes$ and $\QCP$.
\end{theorem}
\begin{proof}
	\Cref{thm:mres-lb} shows that the $\KBKFlq_n$ formula requires
	exponential-size refutations in $\MRes$.
	It has polynomial-size refutations in $\QURes$
	\cite{BWJ-SAT14}, and  also in $\QCP$ (since $\QCP$ simulates $\QURes$
	\cite{BCMS-IC18}).
	Therefore $\MRes$ does not simulate $\QURes$ and $\QCP$.
	The other direction of the non-simulation follows from
	the Equality formulas, as mentioned in \Cref{ex:Equality}. %
\end{proof}

\section{Conclusions and Future Work}
\label{sec:conclusions}
The proof system $\MRes$ was introduced in \cite{DBLP:journals/jar/BeyersdorffBM21}, using the
novel idea of building strategies directly into the proof and using
them to enable additional sound applications of resolution. In
\cite{DBLP:journals/jar/BeyersdorffBM21}, the strengths of the proof system were demonstrated.
In this paper, we complement that study by exposing some limitations
of $\MRes$. We obtain hardness for tree-like $\MRes$ by transferring
computational hardness of the countermodels in decision trees, and for
regular and general $\MRes$ by ad hoc combinatorial arguments.

Several questions still remain.
\begin{enumerate}
\item One of the driving goals behind the definition of $\MRes$ was
  overcoming a perceived weakness of $\LDQRes$: its criterion for
  blocking unsound applications of resolution also blocks several
  sound applications.  However, whether $\MRes$ actually overcomes
  this weakness is not  demonstrated, neither in \cite{DBLP:journals/jar/BeyersdorffBM21} nor here.
  In \cite{DBLP:journals/jar/BeyersdorffBM21},
  $\MRes$ is shown to be more powerful than the reductionless variant
  of $\LDQRes$ (introduced in \cite{BjornerJK15} and further investigated in \cite{PeitlSS19a,DBLP:journals/jar/BeyersdorffBM21}). Very recently, in \cite{MS-SAT22}, this question has been resolved; another variant of $\KBKF$ has been shown to be easy in $\MRes$ but exponentially hard for $\LDQRes$ and even for systems more powerful than $\LDQRes$. 
The other direction, whether there is a formula easy for $\LDQRes$ but hard for $\MRes$, is still open. One possible candidate for this separation might appear to be the original $\KBKF$ formula, which is easy for $\LDQRes$ \cite{ELW13} (that paper uses the name $\varphi_t$). However the KBKF formulas can be shown to have short refutations in $\MRes$ as well, and hence cannot be used for this purpose. Perhaps the completion principle formulas $\CR_n$ may demonstrate this separation. 
\item In the propositional case, regular resolution simulates
  tree-like resolution. This relation may not hold in the case of
  $\MRes$, and even if it does, it will need a different proof.  The
  trick used in the propositional case --- (i)~interpret the proof tree
  as a decision tree for search, (ii)~make the decision tree
  read-once, (iii)~then return from the search tree to a refutation,
  --- does not work here because when we prune away parts of the
  decision tree to get a read-once tree, we may end up destroying
  isomorphism of strategies of blocking variables. Perhaps modifying the proof system itself to require not isomorphism but only semantic equivalence, as was done in \cite{BPS-SAT21},  could lead to a simulation more easily, but that would be in the context of the modified proof system, not $\MRes$ itself.
\end{enumerate}

\section{Acknowledgements}
Olaf Beyersdorff's research was supported by grants from the John Templeton Foundation (grant no.\ 60842), the DFG (BE 4209/3-1), and the Carl Zeiss Foundation. Tom\'{a}\v{s} Peitl's research was supported by Grant J-4361 of the Austrian Science Fund FWF. Olaf Beyersdorff and Meena Mahajan were supported by a DAAD/DST grant. Part of this work was done during the Dagstuhl Seminar `SAT and Interactions' (Seminar 20061).

\bibliography{mergeres,compl}

\end{document}

%% file: notation-acm.tex
\usepackage{amsmath,amsthm}
\usepackage{cleveref}

\usepackage[utf8]{inputenc}
\usepackage{amsfonts}
\usepackage{hyperref}
\usepackage{tikz}
\usetikzlibrary{calc,intersections,decorations.pathmorphing,patterns}
\usepackage{graphicx}
\usepackage{bussproofs}
\usepackage{xcolor}
\usepackage{mathtools}
\usepackage[retainorgcmds]{IEEEtrantools}
\usepackage{relsize}
\usepackage[inline]{enumitem}
\usepackage{thmtools, thm-restate}

\theoremstyle{acmplain}
\newtheorem{theorem}{Theorem}[section]
\newtheorem{claim}[theorem]{Claim}
\newtheorem{observation}[theorem]{Observation}

\newcommand*{\defeq}{\stackrel{\mathsmaller{\mathsf{def}}}{=}}
\renewcommand{\restriction}{\mathord{\upharpoonright}}

\newcommand{\xor}{\text{xor}}
\newcommand{\xorl}{\text{$\xor_l$}}
\newcommand{\QParity}{\text{QParity}}
\newcommand{\LQParity}{\text{LQParity}}
\newcommand{\CR}{\text{CR}}

\newcommand{\KBKF}{\text{KBKF}}
\newcommand{\KBKFlq}{\text{KBKF-lq}}

\newcommand{\Res}{\text{Res}}
\newcommand{\CP}{\text{CP}}
\newcommand{\MRes}{\text{MRes}}

\newcommand{\QRes}{\text{QRes}}
\newcommand{\QURes}{\text{QURes}}
\newcommand{\LDQRes}{\text{LD-QRes}}
\newcommand{\rLDQRes}{\text{rLD-QRes}}
\newcommand{\Exp}{\text{Exp}}
\newcommand{\Red}{\text{Red}}
\newcommand{\ExpRes}{\forall\Exp+\Res}
\newcommand{\QCP}{\CP+\forall\Red}
\newcommand{\IR}{\text{IR}}
\newcommand{\IRM}{\text{IRM}}

\newcommand{\var}{\text{var}}
\newcommand{\vars}{\text{vars}}
\newcommand{\width}{\text{width}}

\newcommand{\res}{\text{res}}

\newcommand{\leaves}{\text{leaves}}
\newcommand{\UsedConstraints}{\text{UC}}
\newcommand{\UsedConstraintInd}{\text{{\sc UCI}}}

\DeclarePairedDelimiter{\card}{\lvert}{\rvert}

%% file: JLBtikz.tex
\usetikzlibrary{arrows}
\usetikzlibrary{decorations.markings}
\usetikzlibrary{calc}

\tikzstyle{download} = [color = gray]
\tikzstyle{axiom} = []
\tikzstyle{inference} = [rounded corners = 0.15cm]
\tikzstyle{conclusion} = []

\tikzstyle{calculus} = [align = center, minimum width=5em, minimum height = 3em, rounded corners = 0.15cm]

\tikzstyle{nicenode} = [draw, align = center, rounded corners = 0.1cm]
\tikzstyle{softrect} = [draw, thick, gray, rounded corners = 0.15cm]

\tikzset{
  on each segment/.style={
    decorate,
    decoration={
      show path construction,
      moveto code={},
      lineto code={
        \path [#1]
        (\tikzinputsegmentfirst) -- (\tikzinputsegmentlast);
      },
      curveto code={
        \path [#1] (\tikzinputsegmentfirst)
        .. controls
        (\tikzinputsegmentsupporta) and (\tikzinputsegmentsupportb)
        ..
        (\tikzinputsegmentlast);
      },
      closepath code={
        \path [#1]
        (\tikzinputsegmentfirst) -- (\tikzinputsegmentlast);
      },
    },
  },
  mid arrow/.style={postaction={decorate,decoration={
        markings,
        mark=at position .5 with {\arrow[#1]{stealth}}
      }}},
}